\documentclass[a4paper,USenglish,cleveref]{lipics-v2021}
\usepackage{xcolor}
\usepackage{todonotes}

\usepackage{multicol}

\usepackage[mode=buildnew]{standalone}

\usepackage{mathtools}

\DeclarePairedDelimiter{\braces}{\{}{\}}            % for braces
\DeclarePairedDelimiter{\parens}{(}{)}              % for parentheses
\DeclarePairedDelimiterX{\setdef}[2]{\{}{\}}{#1 : #2}        % for set builder notation
\DeclarePairedDelimiterXPP{\exclude}[1]{\mathopen{}\setminus}{\{}{\}}{}{#1}

\usepackage{tikz}

\usetikzlibrary{
    decorations,
    decorations.pathmorphing,
    decorations.pathreplacing,
    fit,
    positioning,
    arrows,
    calc,
    backgrounds,
    topaths
}

\tikzstyle{vertex}=[circle, draw, inner sep=2pt,  minimum width=1 pt, minimum size=0.1cm]
\tikzstyle{black_vertex}=[circle, draw, inner sep=2pt, fill=black!100, minimum width=1pt, minimum size=0.1cm]
\tikzstyle{gray_vertex}=[circle, draw, inner sep=2pt, fill=lightgray, minimum width=1pt, minimum size=0.1cm]

\usepackage[bibliography=common]{apxproof} %(conf submission)

\newtheoremrep{theorem}{Theorem} % Do not use with thm-restate
\newtheoremrep{lemma}[theorem]{Lemma}
\usepackage{boxedminipage}
\usepackage[normalem]{ulem}
\usepackage{framed}
\usepackage{xspace}
\usepackage{xifthen}
\usepackage{tabularx}

\newcommand{\bO}{{\cal O}}

\newcommand{\paramstyle}[1]{\text{\rm #1}}

\newcommand{\td}{\paramstyle{td}}
\newcommand{\vc}{\paramstyle{vc}}
\newcommand{\mw}{\paramstyle{mw}}
\newcommand{\ml}{\paramstyle{ml}}
\newcommand{\fvs}{\paramstyle{fvs}}
\newcommand{\vi}{\paramstyle{vi}}
\newcommand{\wvi}{\paramstyle{wvi}}

\newcommand{\fes}{\paramstyle{fes}}

\newcommand{\cc}{\texttt{cc}}

\newcommand{\OPT}{\operatorname{OPT}}

\newcommand{\UBP}{{\sc Unary Bin Packing}}
\newcommand{\BDD}{{\sc Bounded Degree Vertex Deletion}}
\newcommand{\RUBP}{{\sc Restricted Unary Bin Packing}}
\newcommand{\VI}{{\sc Vertex Integrity}}
\newcommand{\UVI}{{\sc Unweighted Vertex Integrity}}
\newcommand{\WVI}{{\sc Weighted Vertex Integrity}}
\newcommand{\SWVI}{{\sc Semi-Weighted Vertex Integrity}}
\newcommand{\COC}{{\sc Component Order Connectivity}}
\newcommand{\SWCOC}{{\sc Semi-Weighted Component Order Connectivity}}

\newcommand{\Z}{\mathbb{Z}}
\newcommand{\N}{\mathbb{N}}

\newlength{\RoundedBoxWidth}
\newsavebox{\GrayRoundedBox}
\newenvironment{GrayBox}[1]{
        \setlength{\RoundedBoxWidth}{.93\textwidth}
        \def\boxheading{#1}
        \begin{lrbox}{\GrayRoundedBox}
        \begin{minipage}{\RoundedBoxWidth}
    }{
        \end{minipage}
        \end{lrbox}
        \begin{center}
            \begin{tikzpicture}%
               \node(Text)[draw=black!20,fill=white,rounded corners,%
                     inner sep=2ex,text width=\RoundedBoxWidth]%
                     {\usebox{\GrayRoundedBox}};
                \coordinate(x) at (current bounding box.north west);
                \node [draw=white,rectangle,inner sep=3pt,anchor=north west,fill=white]
                at ($(x)+(6pt,.75em)$) {\boxheading};
            \end{tikzpicture}
        \end{center}
    }

\newenvironment{defproblemx}[2][]
    {
        \noindent\ignorespaces%
        \FrameSep=6pt%
        \parindent=0pt%
        \vspace*{-1.5em}
        \ifthenelse{\isempty{#1}}{%
            \begin{GrayBox}{\textsc{#2}}%
            }{%
              \begin{GrayBox}{\textsc{#2} parameterized by~{#1}}%
            }
        \begin{tabular*}{\textwidth}{@{\hspace{.1em}} >{\itshape} p{1.8cm} p{0.8\textwidth} @{}}%
    }{
        \end{tabular*}%
        \end{GrayBox}%
        \ignorespacesafterend
    }

\newcommand{\defproblema}[3]{%
    \begin{defproblemx}{#1}
    {\bf Instance:}  & #2 \\
    {\bf Goal:} & #3
    \end{defproblemx}
}%

\newcommand{\problemdef}[3]{\defproblema{#1}{#2}{#3}}

\title{Parameterized Vertex Integrity Revisited}

\titlerunning{Parameterized Vertex Integrity Revisited}

\author{Tesshu Hanaka}{Department of Informatics, Kyushu University, Fukuoka, Japan}{hanaka@inf.kyushu-u.ac.jp}{https://orcid.org/0000-0001-6943-856X}{}
\author{Michael Lampis}{Universit\'{e} Paris-Dauphine, PSL University, CNRS UMR7243, LAMSADE, Paris, France}{michail.lampis@dauphine.fr}{https://orcid.org/0000-0002-5791-0887}{}
\author{Manolis Vasilakis}{Universit\'{e} Paris-Dauphine, PSL University, CNRS UMR7243, LAMSADE, Paris, France}{emmanouil.vasilakis@dauphine.eu}{https://orcid.org/0000-0001-6505-2977}{}
\author{Kanae Yoshiwatari}{Department of Mathematical Informatics, Nagoya University, Aichi, Japan}{yoshiwatari.kanae.w1@s.mail.nagoya-u.ac.jp}{https://orcid.org/0000-0001-5259-7644}{}
\authorrunning{T. Hanaka, M. Lampis, M. Vasilakis, and K. Yoshiwatari}

\Copyright{T. Hanaka, M. Lampis, M. Vasilakis, and K. Yoshiwatari}

%TODO mandatory: Please choose ACM 2012 classifications from https://dl.acm.org/ccs/ccs_flat.cfm 


\category{} %optional, e.g. invited paper

\funding{This work is partially supported by ANR projects ANR-21-CE48-0022 (S-EX-AP-PE-AL) and ANR-18-CE40-0025-01 (ASSK),
and JSPS KAKENHI Grant Numbers JP21K17707, JP21H05852, JP22H00513, and JP23H04388.}

\nolinenumbers %uncomment to disable line numbering

\hideLIPIcs  %uncomment to remove references to LIPIcs series (logo, DOI, ...), e.g. when preparing a pre-final version to be uploaded to arXiv or another public repository

\begin{document}

\maketitle

\begin{abstract}
Vertex integrity is a graph parameter that measures the connectivity of a graph.
Informally, its meaning is that a graph has small vertex integrity if it has a small separator whose removal 
disconnects the graph into connected components which are themselves also small.
Graphs with low vertex integrity are extremely structured; this renders many hard problems tractable and has
recently attracted interest in this notion from the parameterized complexity community.
In this paper we revisit the NP-complete problem of computing the vertex integrity of a given
graph from the point of view of structural parameterizations.
We present a number of new results, which also answer some recently posed open questions from the literature. 
Specifically: We show that unweighted vertex integrity is W[1]-hard parameterized by treedepth;
we show that the problem remains W[1]-hard if we parameterize by feedback edge set size
(via a reduction from a \textsc{Bin Packing} variant which may be of independent interest);
and complementing this we show that the problem is FPT by max-leaf number.
Furthermore, for weighted vertex integrity, we show that the problem admits a single-exponential FPT
algorithm parameterized by vertex cover or by modular width, the latter result improving upon a previous
algorithm which required weights to be polynomially bounded.
\end{abstract}

%%%%%%%%%%%%%%%%%%%%%%%%%%%%%%%%%%%%%%%%%%%%%%%%%%%%%%%%%
% -------------------- INTRODUCTION ---------------------
%%%%%%%%%%%%%%%%%%%%%%%%%%%%%%%%%%%%%%%%%%%%%%%%%%%%%%%%%
\section{Introduction}\label{sec:introduction}

The \emph{vertex integrity} of a graph is a vulnerability measure indicating
how easy it is to break down the graph into small pieces. More precisely,
the vertex integrity $\vi(G)$ of a  graph $G$ is defined as $ \vi(G) =
\min_{S \subseteq V(G)} \{|S| + \max_{D\in \cc(G-S)}|D|\}$, that is, to calculate
the vertex integrity of a graph we must find a separator that minimizes the
size of the separator itself plus the size of the largest remaining connected
component. Intuitively, a graph has low vertex integrity not only when it
contains a small separator, but more strongly when it contains a small
separator such that its removal leaves a collection of small connected
components.

Vertex integrity was first introduced more than thirty years ago by Barefoot et
al.~\cite{Barefoot1987}, but has recently received particular attention from
the parameterized complexity community, because it can be considered as a very
natural structural parameter: when a graph has vertex integrity $k$, large
classes of NP-hard problems admit FPT\footnote{We assume the reader is familiar
with the basics of parameterized complexity, as given e.g. in~\cite{CyganFKLMPPS15}.
We give precise definitions of all parameters in the next section.}
algorithms with running times of the form $f(k)n^{\bO(1)}$~\cite{LampisM21}.
Note that vertex integrity has a clear relationship with
other, well-known structural parameters: it is more restrictive than treedepth,
pathwidth, and treewidth (all these parameters are upper-bounded by vertex
integrity) but more general than vertex cover (a graph of vertex cover $k$ has
vertex integrity at most $k+1$). ``Price of generality'' questions, where one
seeks to discover for a given problem the most general parameter for which an
FPT algorithm is possible, are a central topic in structural parameterized
complexity, and vertex integrity therefore plays a role as a natural stepping
stone in the hierarchy of standard parameters~\cite{BentertHK23,Gaikwad2024a,Gaikwad2024b,GimaHKKO22,GimaO24,LampisM21}.

The investigation of the parameterized complexity aspects of vertex integrity
is, therefore, an active field of research, but it is important to remember
that a prerequisite for any such parameter to be useful is that it should be
tractable to calculate the parameter itself (before we try to use it to solve
other problems). Since, unsurprisingly, computing the vertex integrity exactly
is NP-complete~\cite{Clark1987}, in this paper we focus on this problem from
the point of view of parameterized complexity. We consider both the unweighted,
as well as a natural weighted variant of the problem. Formally, we want to
solve the following:

\problemdef{\textsc{Unweighted (Weighted) Vertex Integrity}}
{A graph $G$ (with binary vertex weights $w : V(G) \to \mathbb{Z}^+$), an integer $k$.}
{Determine whether $\vi(G) \le k$ ($\wvi(G) \le k$).}

The point of view we adopt is that of structural parameterized complexity,
where vertex integrity is the target problem we are trying to solve, and not
necessarily the parameter. Instead, we parameterize by standard structural
width measures, such as variations of treewidth. The questions we would like to
address are of several forms:

\begin{enumerate}
    \item For which structural parameters is it FPT to compute the vertex
    integrity?
    
    \item For which such parameters is it possible to obtain an FPT algorithm with
    single-exponential complexity?
    
    \item For which parameters can the weighted version of the problem be handled
    as well as the unweighted version?
\end{enumerate}

To put these questions in context, we recall some facts from the state of the
art. When the parameter $k$ is the vertex integrity itself, Fellows and
Stueckle show an $\bO(k^{3k}n)$-time algorithm for \UVI~\cite{Fellows1989},
and later Drange et al.~proposed an
$\bO(k^{k+1}n)$-time algorithm even for \WVI~\cite{DrangeDH16},
so this problem is FPT.  More recently, Gima et
al.~\cite{Gima2023} took up the study of vertex integrity in the same
structurally parameterized spirit as the one we adopt here and presented
numerous results which already give some answers to the questions we posed
above. In particular, for the first question they showed that
\UVI{} is W[1]-hard by pathwidth (and hence by
treewidth); for the second question they showed that the problem admits a
single-exponential algorithm for parameter modular-width; and for the third
question they showed that the problem is (weakly) NP-hard on sub-divided stars,
which rules out FPT algorithms for most structural parameters.

% LNCS
% \begin{wrapfigure}{r}{0.40\textwidth}
%     \includegraphics[width=0.9\linewidth]{figures/parameter-comp-unweighted-2.pdf} 
%     \caption{The parameterized complexity of \UVI. The underlined parameters indicate our results.}
%     \label{fig:parameters}
% \end{wrapfigure}

\subparagraph*{Our results.} % LipiCS
% \paragraph*{Our results.} % LNCS
Although the results of~\cite{Gima2023} are rather
comprehensive, they leave open several important questions about the complexity
of vertex integrity. In this paper we resolve the questions explicitly left
open by~\cite{Gima2023} and go on to present several other results that further
clarify the picture for vertex integrity. In particular, our results are as
follows (see also \cref{fig:parameters}):

% LipiCS
\begin{figure}
    \centering
    \includegraphics[width=0.4\textwidth]{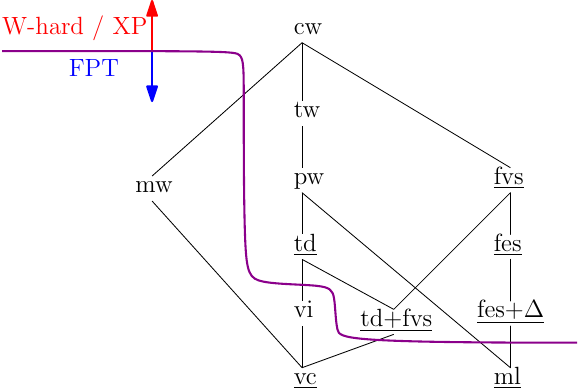}
    \caption{The parameterized complexity of \UVI, with the underlined parameters indicating our results.
    A connection between two parameters implies that the one above generalizes the one below;
    that is, the one below is lower-bounded by a function of the one above.
    All of our FPT algorithms have single-exponential parametric dependence, while the ones for $\vc$ and
    $\mw$ extend to the weighted case as well.}
    \label{fig:parameters}
\end{figure}

The first question we tackle is an explicit open problem from~\cite{Gima2023}:
is \UVI{} FPT parameterized by treedepth? This is a very natural
question, because treedepth is the most well-known parameter that sits between
pathwidth, where the problem is W[1]-hard by~\cite{Gima2023}, and vertex
integrity itself, where the problem is FPT. We resolve this question via a
reduction from \BDD, showing that \UVI{} is W[1]-hard for treedepth (\cref{thm:td}).

A second explicitly posed open question of~\cite{Gima2023} is the complexity of
\UVI{} for parameter feedback vertex set. Taking
a closer look at our reduction from \BDD, which is
known to be W[1]-hard for this parameter, we observe that it also settles this
question, showing that \UVI{} is also hard.
However, in this case we are motivated to dig a little deeper and consider a
parameter, feedback edge set, which is a natural restriction of feedback vertex
set and typically makes most problems FPT. Our second result is to show that
\UVI{} is in fact W[1]-hard even when
parameterized by feedback edge set and the maximum degree of the input graph
(\cref{thm:UVI:fes}).  We achieve this via a reduction from \UBP{} parameterized
by the number of bins, which is W[1]-hard~\cite{JansenKMS13}. An aspect of our
reduction which may be of independent interest is that we use a variant of \UBP{}
where we are given a choice of only two possible bins per item (we observe that
the reduction of~\cite{JansenKMS13} applies to this variant). 

We complement these mostly negative results with a fixed-parameter tractability
result for a more restrictive parameter: we show that \UVI{} is FPT by max-leaf
number (\cref{thm:ml}) indeed by a single-exponential FPT algorithm. Note that
when a graph has bounded max-leaf number, then it has bounded degree and bounded
feedback edge set number, therefore this parameterization is a special case of the
one considered in~\cref{thm:UVI:fes}. Hence, this positive result closely complements
the problem's hardness in the more general case.

Moving on, we consider the parameterization by modular width, and
take a second look at the $2^{\bO(\mw)}n^{\bO(1)}$ algorithm provided
by~\cite{Gima2023}, which is able to handle the weighted case of the problem,
but only for polynomially-bounded weights. Resolving another open problem posed
by~\cite{Gima2023}, we show how to extend their algorithm to handle the general
case of weights encoded in binary (\cref{thm:mw}).

Finally, we ask the question of whether a single-exponential FPT algorithm is
possible for parameters other than max-leaf and modular width. We answer this
affirmatively for vertex cover, even in the weighted case
(\cref{thm:vc:weighted}), obtaining a faster and simpler algorithm for the
unweighted case (\cref{thm:vc:unweighted}).

\subparagraph*{Related work.} % LipiCS
% \paragraph*{Related work.} % LNCS
The concept of vertex integrity is natural enough that it has appeared in many slight variations
under different names in the literature.
We mention in particular, the fracture number~\cite{DvorakEGKO21},
which is the minimum $k$ such that is it possible to delete $k$ vertices
from a graph so that all remaining components have size at most $k$,
and the starwidth~\cite{Ee17}, which is the minumum width of a tree decomposition that is a star.
Both of these are easily seen to be at most a constant factor away from vertex integrity.
Similarly, the safe set number~\cite{BelmonteHKLOO20,FujitaF18} seeks a separator
% \todo{here \& in the following: I felt a bit uncomfortable with the use of the term separator here, as it does not necessarily mean that it is a separator in the usual sense (i.e., a vertex set separating at least two components). It would be good for readers if you add a short note for this.}
such that every component of the separator is only connected to smaller components.
These concepts are so natural that sometimes they are used as parameters without an explicit name,
for example~\cite{Saket24} uses the parameter ``size of a deletion set to a collection of
components of bounded size''.
As observed by~\cite{Gima2023}, despite these similarities, sometimes computing these parameters
can have different complexity, especially when weights are allowed.
Another closely related computational problem, that we also use, is the~\COC~\cite{DrangeDH16} problem,
where we are given explicit distinct bounds on the size of the separator sought and the allowed size of
the remaining components.
\section{Preliminaries}\label{sec:preliminaries}
Throughout the paper we use standard graph notation~\cite{Diestel17}
and assume familiarity with the basic notions of parameterized complexity~\cite{CyganFKLMPPS15}.
All graphs considered are undirected without loops.
Given a graph $G$ and $S \subseteq V(G)$,
$G[S]$ denotes the subgraph induced by $S$, while $G - S$ denotes $G[V(G) \setminus S]$.
If we are additionally given a weight function $w : V(G) \to \mathbb{Z}^+$,
$w(S)$ denotes the sum of the weights of the vertices of $S$, i.e.~$w(S) = \sum_{s \in S} w(s)$.
For $x, y \in \Z$, let $[x, y] = \setdef{z \in \Z}{x \leq z \leq y}$,
while $[x] = [1,x]$.
For a set of integers $S \subseteq \mathbb{Z}^+$, let $\Sigma(S)$ denote the sum of its elements,
i.e.~$\Sigma(S) = \sum_{s \in S} s$, while $\binom{S}{c}$ denotes the set of subsets of $S$ of size $c$,
i.e.~$\binom{S}{c} = \setdef{S' \subseteq S}{|S'| = c}$.

\subsection{Vertex Integrity}

For a vertex-weighted graph $G$ with $w : V(G) \to \mathbb{Z}^+$,
we define its \emph{weighted vertex integrity}, denoted by $\wvi(G)$,
as
\[
    \wvi(G) = \min_{S \subseteq V(G)} \braces*{w(S) + \max_{D \in \cc(G-S)} w(D)},
\]
where $\cc(G-S)$ is the set of connected components of $G-S$.
A set $S$ such that $w(S) + \max_{D \in \cc(G-S)} w(D) \le k$ is called a $\wvi(k)$-set.
The \emph{vertex integrity} of an unweighted graph $G$, denoted by $\vi(G)$,
is defined in an analogous way, by setting $w(v) = 1$ for all $v \in V(G)$.
In that case, $S \subseteq V(G)$ is a $\vi(k)$-set if $|S| + \max_{D\in \cc(G-S)}|D| \leq k$.

% The vertex integrity $\vi(G)$ of a graph $G$ is defined by:
% \[
%     \vi(G) = \min_{S \subseteq V(G)} \braces*{|S| + \max_{D \in \cc(G-S)}|D|},    
% \]
% where $\cc(G-S)$ is the family of vertex sets of connected components of $G-S$.
% A set $S$ such that $|S| + \max_{D\in \cc(G-S)}|D|\le k$ is called a $\vi(k)$-set.
% Analogously, for a vertex-weighted graph with vertex weight function $w : V(G) \to \mathbb{Z}^+$,
% the weighted vertex integrity is defined by:
% \[
%     \wvi(G) = \min_{S \subseteq V(G)} \braces*{w(S) + \max_{D \in \cc(G-S)} w(D)},
% \]
% where $w(D) = \sum_{v \in D} w(v)$.
% A set $S$ such that $w(S) + \max_{D \in \cc(G-S)} w(D) \le k$ is called a $\wvi(k)$-set.

A vertex $v \in S$ is called \emph{redundant} if at most one connected component of $G-S$ contains neighbors of $v$.
A set $S \subseteq V(G)$ is \emph{irredundant} if $S$ contains no redundant vertex.
Thanks to~\cref{prop:irredundant}, it suffices to only search for irredundant $\wvi(k)$-sets when solving \VI.

\begin{proposition}[\cite{DrangeDH16,Gima2023}]\label{prop:irredundant}
    A graph with a $\wvi(k)$-set has an irredundant $\wvi(k)$-set.
\end{proposition}

\subsection{Graph Parameters}

We use several standard graph parameters, so we recall here their definitions
and known relations between them. A graph $G$ has \emph{feedback vertex} (respectively
\emph{edge}) \emph{set} $k$ if there exists a set of $k$ vertices (respectively edges) such
that removing them from $G$ destroys all cycles. We use $\fvs(G)$ and $\fes(G)$
to denote these parameters. Note that even though computing $\fvs(G)$ is
NP-complete~\cite{Karp72}, in all connected graphs with $m$ edges and $n$ vertices
$\fes(G)=m-n+1$. The \emph{vertex cover} of a graph $G$, denoted by $\vc(G)$, is the
size of the smallest set whose removal destroys all edges.
The \emph{treedepth} of a graph $G$ can be defined recursively as follows: $\td(K_1)=1$; if $G$ is
disconnected $\td(G)$ is equal to the maximum of the treedepth of its
connected components; otherwise $\td(G)=\min_{v\in V(G)} \td(G-v)+1$. The
\emph{max-leaf number} of a graph $G$, denoted by $\ml(G)$, is the maximum number of
leaves of any spanning tree of $G$. 

A module of a graph $G=(V,E)$ is a set of vertices $M \subseteq V$ such that for
all $x \in V \setminus M$ we have that $x$ is either adjacent to all vertices of
$M$ or to none. The \emph{modular width} of a graph $G=(V,E)$
(\cite{GajarskyLMMO15,GajarskyLO13}) is the smallest integer $k$ such that,
either $|V| \le k$, or $V$ can be partitioned into at most $k' \leq k$ sets $V_1,\ldots,
V_{k'}$, with the following two properties: (i) for all
$i \in [k']$, $V_i$ is a module of $G$,
(ii) for all $i \in [k']$, $G[V_i]$ has modular
width at most $k$.

Let us also briefly explain the relations depicted in~\cref{fig:parameters}.
Clearly, for all $G$, we have $\fvs(G) \le \fes(G)$, because
we can remove from the graph one endpoint of each edge of the feedback edge
set. It is known that if a graph has $\ml(G)=k$, then $G$ contains at most
$\bO(k)$ vertices of degree $3$ or more (\cref{lem:kw}), and clearly such a graph
has maximum degree at most $k$. Since vertices of degree at most $2$ are
irrelevant for $\fes$, we conclude that the parameterization by $\ml$ is more
restrictive than that for $\fes+\Delta$. It is also not hard to see that for
all $G$, $\td(G) \le \vi(G) \le \vc(G)+1$. Note also that even though $\vc$ can
be seen as a parameter more restrictive than $\mw$, when a graph has vertex
cover $k$, the best we can say is that its modular width is at most $2^k+k$
\cite{Lampis12}. As a result, the algorithm of \cref{thm:mw} does not imply a
single-exponential FPT algorithm for parameter $\vc$ (but does suffice to show
that the problem is FPT). We also note that the reductions of \cref{thm:td}
(for $\td$) and \cref{thm:UVI:fes} (for $\fes+\Delta$) are complementary and
cannot be subsumed by a single reduction. The reason for this is that if in a
graph we bound simultaneously the treedepth and the maximum degree, then we
actually bound the size of the graph (rendering all problems FPT).

\section{Treedepth}\label{sec:td}

Our main result in this section is the following theorem, resolving a question of~\cite{Gima2023}.
We obtain it via a parameter-preserving reduction from \BDD,
which is known to be W[1]-hard parameterized by treedepth plus feedback vertex set~\cite{GanianKO21}.

\begin{theorem}\label{thm:td}
    \UVI{} is W[1]-hard parameterized by $\td + \fvs$.
    Moreover, it cannot be solved in time $f(\td) n^{o(\td)}$ under the ETH.
\end{theorem}

\begin{proof}
    First we define the closely related \COC{} problem:
    given a graph $G$ as well as integers $\ell$ and $p$,
    we want to determine whether there exists $S \subseteq V(G)$ such that $|S| \leq p$ and
    all components of $G - S$ have size at most $\ell$.
    We will proceed in two steps:
    we first reduce \BDD{} to \COC, and then employ the reduction of~\cite{Gima2023} that reduces the latter to \UVI.
    Notice that~\cite[Lemma~4.4]{Gima2023} creates an equivalent instance of \UVI{} by solely adding disjoint stars and leaves in the vertices of the initial graph, therefore it suffices to prove the statement
    for \COC{} instead.
    
    We give a parameterized reduction from \BDD, which is W[1]-hard by treedepth plus feedback vertex set number~\cite{GanianKO21}
    and cannot be solved in time $f(\td) n^{o(\td)}$ under the ETH~\cite{LampisV23}.
    In \BDD{} we are given a graph $G=(V,E)$ and two integers $k$ and $d$,
    and we are asked to determine whether there exists $S \subseteq V$ of size $|S| \leq k$
    such that the maximum degree of $G-S$ is at most $d$.
    In the following, let $n = |V(G)|$ and $m = |E(G)|$.

    Given an instance $(G, k, d)$ of \BDD, we construct an equivalent instance $(G', \ell, p)$ of \COC.
    We construct $G'$ from $G$ as follows:
    We subdivide every edge $e = \{u,v\} \in E(G)$ three times,
    thus replacing it with a path on vertices $u$, $u_v$, $y_e$, $v_u$, and $v$,
    where $T_e = \{u_v, y_e, v_u\}$.
    Next, we attach $d-1$ leaves to $y_e$ 
    (see~\cref{fig:edge_subdivide}).
    This concludes the construction of $G'$.
    Notice that the subdivision of the edges three times and the attachment of pendant vertices
    does not change the feedback vertex set number, while the treedepth is only increased by an additive constant.
    Thus, it holds that $\fvs(G') = \fvs(G)$ and $\td(G') = \td(G) + \bO(1)$.
    
    \begin{figure}[htbp]
        \centering
        \includegraphics[height=0.2\linewidth]{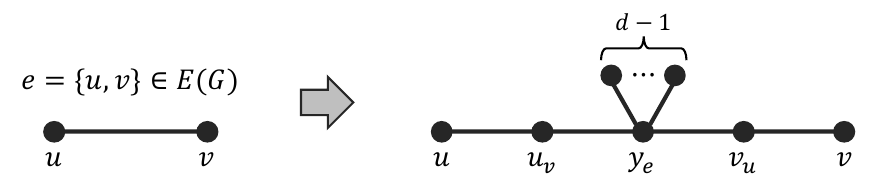}
        \caption{Edge gadget for edge $e = \{u,v\} \in E(G)$.}
        \label{fig:edge_subdivide}
    \end{figure}
    
    In the following, we show that $(G,k,d)$ is a yes-instance of \BDD{} if and only if
    $(G', \ell, p)$ is a yes-instance of \COC, where $\ell = d + 1$ and $p = k + m$.
    
    For the forward direction, let $S$ be a set of vertices of size at most $k$ such that
    the maximum degree of $G-S$ is at most $d$.
    We will construct a set $S' \subseteq V(G')$ such that $|S'| \leq p$ and every connected component
    of $G' - S'$ has size at most $\ell$.
    Initially set $S' = S$.
    Then, add one vertex to $S'$ per edge $e = \{u,v\} \in E(G)$ as follows.
    If $u,v \in S$ or $u,v \notin S$, we add $y_e$ to $S'$.
    Otherwise, if $u \in S$ and $v \notin S$, we add $v_u$ to $S'$;
    symmetrically, if $u \notin S$ and $v \in S$, we add $u_v$ instead (see~\cref{fig:separator}).
    Notice that $|S'| = |S| + m \leq k + m = p$,
    therefore it suffices to show that the size of each connected component of $G'-S'$
    is at most $\ell = d + 1$.
    
    Consider a connected component $D$ of $G'-S'$.
    Assume that $D$ does not contain any vertices of $V \setminus S$.
    If $D$ is a leaf it holds that $|D| \leq d + 1$.
    Alternatively, $D$ is a subgraph of the graph induced by $u_v$ (or $v_u$),
    $y_e$, and its attached leaves,
    for some $e = \{u,v\} \in E(G)$, in which case $|D| \leq d + 1$.
    Now assume that $D$ contains $u \in V\setminus S$.
    Notice that $u$ is the only vertex of $V \setminus S$ present in $D$,
    since $S' \cap T_e \neq \emptyset$ for all $e \in E(G)$.
    Moreover, let $N(u) \setminus S = \setdef{u_i}{i \in [q]}$ denote its neighbors in $G-S$,
    where $q \leq d$ since the maximum degree of $G-S$ is at most $d$.
    In that case, it follows that $D$ consists of $u$,
    as well as the vertices $u_{u_i}$ for all $i \in [q]$.
    Consequently, $|D| = q + 1 \leq d + 1$.

    \begin{figure}[htbp]
        \centering
        \includegraphics[height=0.38\linewidth]{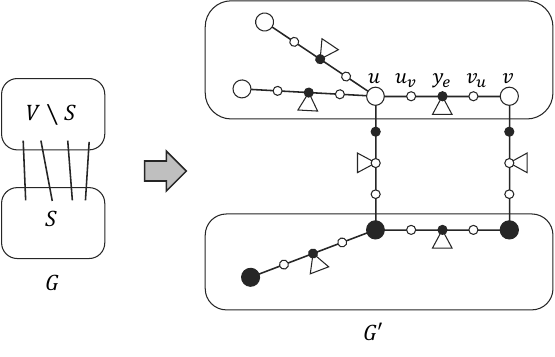}
        \caption{Black vertices belong to $S'$.}
        \label{fig:separator}
    \end{figure}

    For the converse direction, assume there exists $S' \subseteq V(G')$ such that $|S'| \leq p = k + m$
    and $|D| \leq \ell = d + 1$, for all connected components $D \in \cc(G'-S')$.
    Assume that $S'$ does not contain any leaves; if it does, substitute them with their single neighbor.
    Moreover, $S' \cap T_e \neq \emptyset$ for all $e \in E(G)$,
    since otherwise $G' - S'$ has a component of size at least $d + 2 > \ell$,
    which is a contradiction.
    Assume without loss of generality that $|S' \cap T_e| = 1$, for all $e = \{u, v \} \in E(G)$;
    if that is not the case, there is always a vertex of $\braces{u_v, v_u}$,
    say $u_v$, such that $u_v \in S'$ and $S' \cap \braces{y_e, v_u} \neq \emptyset$,
    in which case one may consider the deletion set $(S' \cup \{u\}) \setminus \{u_v\}$ instead
    (the argument is symmetric in case $v_u \in S'$).
    
    Let $S = S' \cap V$, where $|S| \leq k$.
    We will prove that $G-S$ has maximum degree at most $d$.
    Let $D_u$ denote the connected component of $G'-S'$ that contains $u \in V \setminus S'$;
    in fact this is the only vertex of $V \setminus S'$ present in $D_u$,
    since $S' \cap T_e \neq \emptyset$ for all $e \in E(G)$.
    Notice that for all $e = \{u,v\} \in E(G)$ where $u,v \notin S'$,
    it holds that $y_e \in S'$:
    if that were not the case, then either $D_u$ or $D_v$ contains at least $d + 2 > \ell$ vertices,
    due to $\{u, u_v, y_e\}$ or $\{v, v_u, y_e\}$ and the leaves of $y_e$ respectively.
    For $u \in V \setminus S$, let $N(u) \setminus S = \setdef{u_i}{i \in [q]}$, for some integer $q$,
    denote its neighbors in $G-S$,
    where $e_i = \{u,u_i\} \in E(G)$ for $i \in [q]$.
    It suffices to show that $q \leq d$.
    Assume that this is not the case, i.e.~$q > d$.
    Then, since $S' \cap T_{e_i} = \{y_{e_i}\}$ for $i \in [q]$,
    it follows that $D_u$ contains vertices $u$ and $u_{u_i}$,
    therefore $|D_u| \geq q + 1 > d + 1 = \ell$, which is a contradiction.
    Consequently, $|N(u) \setminus S| \leq d$ for all $u \in V \setminus S$,
    i.e.~$G - S$ has maximum degree $d$.
\end{proof}

\section{Feedback Edge Set plus Maximum Degree}\label{sec:fes}

In this section we prove that \VI{} is W[1]-hard parameterized by $\fes+\Delta$.
Since our reduction is significantly more involved than the one of~\cref{thm:td},
we proceed in several steps. We start from an instance of \UBP{} where the parameter
is the number of bins and consider a variant where we are also supplied in the input,
for each item, a choice of two possible bins to place it. We first observe that the
reduction of~\cite{JansenKMS13} shows that this variant is also W[1]-hard. We then
reduce this to a \emph{semi-weighted} version of \VI, where placing a vertex in the
separator always costs $1$, but vertices have weights which they contribute to their
components if they are not part of the separator, and where we are prescribed the size
of the separator to use (this is called the \COC{} problem).
Subsequently, we show how to remove the weights and the prescription on the separator size to obtain
hardness for \VI.

\subsection{Preliminary Tools}\label{subsec:fes_tools}

\subparagraph*{Unary Bin Packing.} % LipiCS
% \paragraph*{Unary Bin Packing.} % LNCS
Given a set $S = \braces{s_1, \ldots, s_n}$ of integers in unary (i.e.~$s_i = \bO(n^c)$ for some constant $c$),
as well as $k \in \mathbb{Z}^+$, \UBP{} asks whether we can partition $S$ into $k$ subsets $S_1, \ldots, S_k$,
such that $\Sigma(S_i) = \Sigma(S) / k$, for all $i \in [k]$.
This problem is well known to be W[1]-hard parameterized by the number of bins $k$~\cite{JansenKMS13}.
We formally define a restricted version where every item is allowed to choose between \emph{exactly two} bins,
and by delving deeper into the proof of~\cite{JansenKMS13} we observe that an analogous hardness result follows.

\problemdef{\RUBP}
{A set $S = \braces{s_1, \ldots, s_n}$ of integers in unary, $k \in \mathbb{Z}^+$,
as well as a function $f : S \to \binom{[k]}{2}$.}
{Determine whether we can partition $S$ into $k$ subsets $S_1, \ldots, S_k$,
such that for all $i \in [k]$ it holds that
(i) $\Sigma(S_i) = \Sigma(S) / k$,
and (ii) $\forall s \in S_i$, $i \in f(s)$.}

\begin{theoremrep}
    \RUBP{} is W[1]-hard parameterized by the number of bins.
\end{theoremrep}

\begin{proof}
    The W[1]-hardness of \UBP{} parameterized by the number of bins is shown via an intermediate problem,
    called \textsc{$10$-Unary Vector Bin Packing}~\cite{JansenKMS13}.    
    In said problem, we are given $n$ items $\mathcal{S} = \braces{\mathbf{s}_1, \ldots, \mathbf{s}_n}$,
    where every item is a 10-dimensional vector belonging to $\N^{10}$ encoded in unary.
    Additionally, we are given $k$ ``bin'' vectors $\mathcal{B} = \braces{\mathbf{B}_1, \ldots, \mathbf{B}_k}$,
    and the question is whether we can partition the items into $k$ sets $J_1, \ldots, J_k$ such that
    $\sum_{\mathbf{s} \in J_i} \mathbf{s} \leq \mathbf{B}_i$, for all $i \in [k]$.

    This problem is known to be W[1]-hard parameterized by the number of bins,
    via an fpt-reduction from \textsc{Subgraph Isomorphism} parameterized by the number of edges of the sought subgraph.
    Fix $1 \leq i < j \leq k$ and notice that for the intended solution of
    said reduction (\cite[Lemma~10]{JansenKMS13}) it holds that:
    \begin{itemize}
        \item every item $\mathbf{s}_{i,j}(e)$ is either placed in bin $\mathbf{P}_{i,j}$ or bin $\mathbf{R}$,
        \item every item $\mathbf{t}_{i,i}(v)$ is either placed in bin $\mathbf{Q}_{i}$ or bin $\mathbf{R}$,
        \item every item $\mathbf{t}_{i,j}(v)$ and $\mathbf{t}_{j,i}(v)$ is either placed in bin $\mathbf{P}_{i,j}$ or bin $\mathbf{Q}_i$.
    \end{itemize}
    Consequently, the hardness result holds in the case when every item is allowed to be placed in one amongst two specified bins, called the \textsc{$10$-Unary Restricted Vector Bin Packing} problem.

    As a second step, the authors reduce \textsc{$10$-Unary Vector Bin Packing} to
    \UBP{} (\cite[Lemma~6]{JansenKMS13}).
    Assume that $(\mathcal{S}, \mathcal{B}, f)$
    denotes the initial instance of \textsc{$10$-Unary Restricted Vector Bin Packing},
    where $f(\mathbf{s}) \in \binom{\mathcal{B}}{2}$ denotes the bins an item $\mathbf{s} \in \mathcal{S}$ may be placed into.
    To this end, the authors introduce a bin $B_i$ per bin $\mathbf{B}_i$, an item $s_i$ per item $\mathbf{s}_i$,
    as well as $k$ additional items $t_1, \ldots, t_k$,
    which are used to encode the capacity of bins $\mathbf{B}_i$ and are
    sufficiently large to guarantee that no two of them are placed in the same bin.
    It suffices to slightly modify said reduction, in order to show hardness for \RUBP.
    In particular, we introduce an additional copy $s'_i$ per item $s_i$, an additional copy $t'_i$ per item $t_i$,
    as well as an additional copy $B'_i$ per bin $B_i$.
    Next, if $f(\mathbf{s}) = \braces{\mathbf{B}_i, \mathbf{B}_j}$,
    we set $f'(s) = \braces{B_i, B_j}$ and $f'(s') = \braces{B'_i, B'_j}$.
    Moreover, we set $f'(t_i) = f'(t'_i) = \braces{B_i, B'_i}$.
    Since all items $t_i$ and $t'_i$ are placed in distinct bins, the correctness follows as in the proof of~\cite{JansenKMS13}.
\end{proof}

\subparagraph*{Semi-weighted problems.} % LipiCS
% \paragraph*{Semi-weighted problems.} % LNCS
In this section we study semi-weighted versions of \COC{} and \VI,
which we first formally define.
Then, we prove that the first can be reduced to the latter,
while retaining the size of the minimum feedback edge set and the maximum degree.

\problemdef{\SWCOC}
{A vertex-weighted graph $G = (V, E, w)$, as well as integers $\ell, p \in \mathbb{Z}^+$.}
{Determine whether there exists $S \subseteq V$ of size $|S| \leq p$,
such that $w(D) \leq \ell$ for all $D \in \cc(G-S)$.}

\problemdef{\SWVI}
{A vertex-weighted graph $G = (V, E, w)$, as well as an integer $\ell \in \mathbb{Z}^+$.}
{Determine whether there exists $S \subseteq V$ such that
$|S| + w(D) \leq \ell$ for all $D \in \cc(G-S)$.}

\begin{theoremrep}\label{thm:swcoc_to_swvi}
    \SWCOC{} parameterized by $\fes + \Delta$ is fpt-reducible to \SWVI{} parameterized by $\fes + \Delta$.
\end{theoremrep}

\begin{proof}
    We will closely follow the proof of~\cite[Lemma~4.4]{Gima2023}.
    Let $(G, w, \ell, p)$ be an instance of \SWCOC.
    Assume without loss of generality that for every vertex $v \in V(G)$ it holds that $w(v) \leq \ell$,
    since otherwise it necessarily belongs to the deletion set.
    We will construct an equivalent instance $(G', w', k)$ of \SWVI,
    where $k = \ell p + \ell + p$.
    Construct $G'$ in the following way:
    Make a copy of $G$, where $w'(v) = w(v)$, for all $v \in V(G)$.
    Then, attach to every vertex $v$ a leaf $l_v$, and set $w'(l_v) = p \cdot w(v)$.
    Finally, introduce an independent set $I = \setdef{v_i}{i \in [k + 1]}$,
    with every vertex of which having weight $w'(v_i) = k - p = \ell p + \ell$.
    This concludes the construction of the instance.
    
    For the forward direction, assume there exists $S \subseteq V(G)$ such that $|S| \leq p$ and $w(D) \leq \ell$,
    for all connected components $D$ of $G-S$.
    It suffices to prove that $w(D') \leq k - p = \ell p + \ell$, for all connected components $D'$ of $G'-S$.
    If $D' \cap V(G) = \emptyset$, then $D'$ contains
    (i) either a single vertex $l_v$ of weight $w'(l_v) = p \cdot w(v) \leq p \ell$,
    (ii) or a single vertex $v_i$ of weight $w'(v_i) = \ell p + \ell$,
    and the statement holds.
    Alternatively, $D' \cap V(G)$ induces a connected component $D$ of $G-S$,
    therefore $w(D') \leq w(D) + p \cdot w(D) = \ell + \ell p$.
    
    For the converse direction, assume there exists $S' \subseteq V(G')$ such that
    $|S'| + w'(D') \leq k$, for all connected components $D'$ of $G'-S'$.
    Assume without loss of generality that $S'$ does not contain any vertex of degree $1$;
    if that is the case, substitute it with its single neighbor and this set remains a valid solution.
    Notice that $|I| = k+1 > |S'|$, and let $v_i \in I$ belong to $G'-S'$.
    In that case, it follows that $\max_{D' \in \cc(G'-S')} w'(D') \geq w'(v_i) = k-p$,
    where $\cc(G'-S')$ denotes the connected components of $G'-S'$.
    Consequently, it follows that $|S'| \leq p$.
    Let $S = S' \cap V(G)$ and $D$ be an arbitrary connected component of $G-S$.
    Since $|S| \leq |S'| \leq p$, it suffices to prove that $w(D) \leq \ell$.
    Let $D'$ be the connected component of $G'-S'$ such that $D \subseteq D'$.
    Since $S'$ does not contain any vertices of degree 1,
    it holds that for each vertex of $D$, $D'$ contains the corresponding leaf attached to it,
    and thus, $w(D') = w(D) + p \cdot w(D) = (p+1) w(D)$.
    Since $w(D') \leq k = \ell p + \ell + p < (p+1) (\ell+1)$,
    it follows that $w(D) < \ell+1$.

    Finally, to conclude the proof, notice that $\fes(G') = \fes(G)$ as well as $\Delta(G') \leq \Delta(G) + 1$,
    since the only vertices added are one leaf per vertex and vertices of degree 0.
\end{proof}

\subsection{Hardness Result}\label{subsec:rubp_to_swcoc}

Using the results of~\cref{subsec:fes_tools}, we proceed to proving the main theorem of this section.
To this end, we present a reduction from \RUBP{} to \SWCOC{} such that for the produced graph $G$ it holds that $\fes(G) + \Delta(G) \leq f(k)$,
for some function $f$ and $k$ denoting the number of bins of the \RUBP{} instance.

We first provide a sketch of our reduction.
For every bin of the \RUBP{} instance, we introduce a clique of $\bO(k)$ heavy vertices,
and then connect any pair of such cliques via two paths.
The weights are set in such a way that an optimal solution will only delete vertices from said paths.
In order to construct a path for a pair of bins, we compute the set of all subset sums of the items that can be placed in these two
bins, and introduce a vertex of medium weight per such subset sum.
Moreover, every such vertex corresponding to subset sum $s$ is preceded by exactly $s$ vertices of weight $1$.
An optimal solution will cut the path in such a way that the number of vertices of weight $1$ will be partitioned between the two bins,
encoding the subset sum of the elements that are placed on each bin.
The second path that we introduce has balancing purposes, allowing us to exactly count the number of vertices of medium weight that every
connected component will end up with.

\begin{theoremrep}\label{thm:SWCOC:fes}
    \SWCOC{} is W[1]-hard parameterized by $\fes + \Delta$.
\end{theoremrep}

\begin{proof}
    Let $(A, k, f)$ be an instance of \RUBP,
    where $A = \braces{a_1, \ldots, a_n}$ denotes the set of items given in unary,
    $k$ is the number of bins, and
    $f : A \to \binom{[k]}{2}$ dictates the bins an item may be placed into.
    In the following, consider $B = \Sigma(A) / k$,
    as well as $M = kB+1$ and $L = 8k^2 BM$.
    Notice that $M > kB$,
    while $L / (4k) \in \N$ and $L/(4k) > (k-1) \cdot 2B M + B$.
    We will reduce $(A,k,f)$ to an equivalent instance $(G,w,\ell,p)$ of \SWCOC,
    where $\ell = L + (k-1) \cdot 2B M + B$ and $p = 3 \binom{k}{2}$ denote the maximum component weight and size of the deletion set respectively.
    
    For every $i \in [k]$, we introduce a clique on vertex set $\hat{C}_i$,
    which is comprised of $4k$ vertices, each of weight $L / (4k)$.
    Fix $i$ and $j$ such that $1 \leq i < j \leq k$, and let $H_{i,j} = \setdef{a \in A}{f(a) = \braces{i,j}}$ denote the subset of items
    which can be placed either on bin $i$ or bin $j$, where $\Sigma(H_{i,j}) \leq 2B$.
    Let $\mathcal{S}(H_{i,j}) = \setdef{\Sigma(H)}{H \subseteq H_{i,j}}$ denote the set of all subset sums of $H_{i,j}$,
    and notice that since every element of $H_{i,j}$ is in unary,
    $\mathcal{S}(H_{i,j})$ can be computed in polynomial time using e.g.~Bellman's classical DP algorithm~\cite{Bellman}.

    Next, we will construct two paths connecting the vertices of $\hat{C}_i$ and $\hat{C}_j$.
    First, introduce vertex set $\hat{U}_{i,j} = \setdef{v^{i,j}_q}{q \in [0, 4B - |\mathcal{S}(H_{i,j})| +1]}$,
    where $w(v) = M$ for all $v \in \hat{U}_{i,j}$.
    Add edges $(v^{i,j}_q, v^{i,j}_{q+1})$ for all $q \in [0, 4B - |\mathcal{S}(H_{i,j})|]$,
    as well as $(v_1, v^{i,j}_0)$ and $(v^{i,j}_{4B - |\mathcal{S}(H_{i,j})| +1}, v_2)$,
    for all $v_1 \in \hat{C}_i$ and $v_2 \in \hat{C}_j$.
    Next, introduce vertices
    \begin{itemize}
        \item $\setdef{s^{i,j}_q, t^{i,j}_q}{q \in [\Sigma(H_{i,j})]}$,
        each of weight $1$,
        
        \item $\setdef{\sigma^{i,j}_q, \tau^{i,j}_q}{q \in \mathcal{S}(H_{i,j})}$,
        each of weight $M$,
        
        \item $\hat{D}^3_{i,j} = \braces{b^{i,j}_1, b^{i,j}_2, b^{i,j}_3}$,
        where $w(b^{i,j}_1) = w(b^{i,j}_3) = L/2$,
        and $w(b^{i,j}_2) = (k-1) \cdot 2BM + B - \Sigma(H_{i,j}) - (|\mathcal{S}(H_{i,j})| - 1) \cdot M$,
    \end{itemize}
    and set $\hat{D}_{i,j} = \hat{D}^1_{i,j} \cup \hat{D}^2_{i,j} \cup \hat{D}^3_{i,j}$,
    where $\hat{D}^1_{i,j} = \setdef{s^{i,j}_q}{q \in [\Sigma(H_{i,j})]} \cup \setdef{\sigma^{i,j}_q}{q \in \mathcal{S}(H_{i,j})}$
    and $\hat{D}^2_{i,j} = \setdef{t^{i,j}_q}{q \in [\Sigma(H_{i,j})]} \cup \setdef{\tau^{i,j}_q}{q \in \mathcal{S}(H_{i,j})}$.
    Then, add the following edges:
    \begin{itemize}
        \item $(v_1, \sigma^{i,j}_0)$ and $(\tau^{i,j}_{\Sigma(H_{i,j})}, v_2)$, for all $v_1 \in \hat{C}_i$ and $v_2 \in \hat{C}_j$,
        \item $(\sigma^{i,j}_{\Sigma(H_{i,j})}, b^{i,j}_1)$, $(b^{i,j}_1, b^{i,j}_2)$, $(b^{i,j}_2, b^{i,j}_3)$ and $(b^{i,j}_3, \tau^{i,j}_0)$,
        \item for $q \in \mathcal{S}(H_{i,j})$,
        add edges $(s^{i,j}_q, \sigma^{i,j}_q)$ and $(t^{i,j}_q, \tau^{i,j}_q)$ if $q \neq 0$,
        and edges $(\sigma^{i,j}_q, s^{i,j}_{q+1})$ and $(\tau^{i,j}_q, t^{i,j}_{q+1})$ if $q \neq \Sigma(H_{i,j})$,
        \item for $q \in [\Sigma(H_{i,j})] \setminus \mathcal{S}(H_{i,j})$,
        add edges $(s^{i,j}_q, s^{i,j}_{q+1})$ and $(t^{i,j}_q, t^{i,j}_{q+1})$.
    \end{itemize}

    This concludes the construction of $G$.
    See~\cref{fig:fes_construction} for an illustration.

    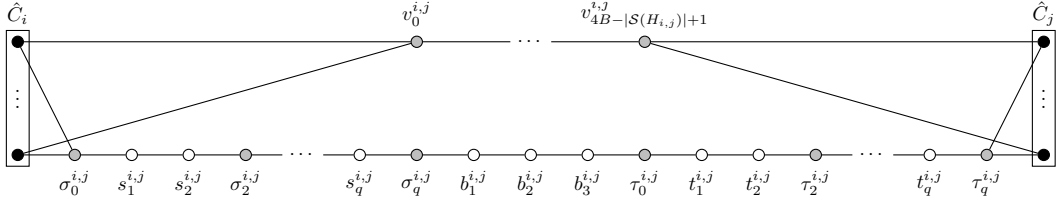
\begin{figure}[ht]
\centering

\begin{tikzpicture}[scale=0.75, transform shape]

% Left Clique
\node[black_vertex] (c11) at (5,5) {};
\node[] () at (5,6.1) {$\vdots$};
\node[black_vertex] (c1n) at (5,7) {};
\draw[] (4.8,4.8) rectangle (5.2,7.2);
\node[] () at (5,7.5) {$\hat{C}_i$};

% Right Clique
\begin{scope}[shift={(18,0)}]
    \node[black_vertex] (c21) at (5,5) {};
    \node[] () at (5,6.1) {$\vdots$};
    \node[black_vertex] (c2n) at (5,7) {};
    \draw[] (4.8,4.8) rectangle (5.2,7.2);
    \node[] () at (5,7.5) {$\hat{C}_j$};
\end{scope}

% Upper Path
\begin{scope}[shift={(7,-3)}]
    \node[gray_vertex] (p0) at (5,10) {};
    \node[] () at (5,10.5) {$v^{i,j}_0$};
    \node[] (temp) at (7,10) {$\cdots$};
    \node[gray_vertex] (pB) at (9,10) {};
    \node[] () at (9,10.5) {$v^{i,j}_{4B - |\mathcal{S}(H_{i,j})| + 1}$};
    \draw[] (p0)--(temp)--(pB);
    \draw[] (p0)--(c11);
    \draw[] (p0)--(c1n);
    \draw[] (pB)--(c21);
    \draw[] (pB)--(c2n);
\end{scope}

% Lower Path
\begin{scope}[shift={(1,0)}]
    \node[gray_vertex] (sigma0) at (5,5) {};
    \node[] () at (5,4.5) {$\sigma^{i,j}_0$};
    \node[vertex] (s1) at (6,5) {};
    \node[] () at (6,4.5) {$s^{i,j}_1$};
    \node[vertex] (s2) at (7,5) {};
    \node[] () at (7,4.5) {$s^{i,j}_2$};
    \node[gray_vertex] (sigma2) at (8,5) {};
    \node[] () at (8,4.5) {$\sigma^{i,j}_2$};
    \node[] (temp2) at (9,5) {$\cdots$};
    \node[vertex] (sB) at (10,5) {};
    \node[] () at (10,4.5) {$s^{i,j}_{q}$};
    \node[gray_vertex] (sigmaB) at (11,5) {};
    \node[] () at (11,4.5) {$\sigma^{i,j}_{q}$};
    \node[vertex] (big1) at (12,5) {};
    \node[] () at (12,4.5) {$b^{i,j}_{1}$};
    \node[vertex] (big2) at (13,5) {};
    \node[] () at (13,4.5) {$b^{i,j}_{2}$};
    \node[vertex] (big3) at (14,5) {};
    \node[] () at (14,4.5) {$b^{i,j}_{3}$};
    \node[gray_vertex] (tau0) at (15,5) {};
    \node[] () at (15,4.5) {$\tau^{i,j}_{0}$};
    \node[vertex] (t1) at (16,5) {};
    \node[] () at (16,4.5) {$t^{i,j}_{1}$};
    \node[vertex] (t2) at (17,5) {};
    \node[] () at (17,4.5) {$t^{i,j}_{2}$};
    \node[gray_vertex] (tau2) at (18,5) {};
    \node[] () at (18,4.5) {$\tau^{i,j}_{2}$};
    \node[] (temp3) at (19,5) {$\cdots$};
    \node[vertex] (tB) at (20,5) {};
    \node[] () at (20,4.5) {$t^{i,j}_{q}$};
    \node[gray_vertex] (tauB) at (21,5) {};
    \node[] () at (21,4.5) {$\tau^{i,j}_{q}$};
    
    \draw[] (sigma0)--(s1)--(s2)--(sigma2)--(temp2)--(sB)--(sigmaB)--(big1)--(big2)--(big3)--(tau0)--(t1)--(t2)--(tau2)--(temp3)--(tB)--(tauB);
    \draw[] (sigma0)--(c11);
    \draw[] (sigma0)--(c1n);
    \draw[] (tauB)--(c21);
    \draw[] (tauB)--(c2n);
\end{scope}

\end{tikzpicture}

\caption{Rectangles denote cliques of size $4k$.
Here we assume that $1 \leq i < j \leq k$, $2 \in \mathcal{S}(H_{i,j})$,
while $q = \Sigma(H_{i,j})$.
It holds that $w(b^{i,j}_1) = w(b^{i,j}_3) = L/2$, while $w(b^{i,j}_2) = (k-1) \cdot 2BM + B - \Sigma(H_{i,j}) - (|\mathcal{S}(H_{i,j})| - 1) \cdot M$.
For the rest of the vertices, the white, gray,
or black color indicates weight of $1$, $M$, or $L / (4k)$ respectively.}
\label{fig:fes_construction}

\end{figure}

    \begin{lemma}\label{lemma:fes_forward_direction}
        If $(A,k,f)$ is a Yes-instance of \RUBP, then $(G,w,\ell,p)$ is a Yes-instance of \SWCOC.
    \end{lemma}

    \begin{nestedproof}
        Let $(\mathcal{A}_1, \ldots, \mathcal{A}_k)$ be a partition of $A$ such that for all $i \in [k]$ it holds that
        (i) $\Sigma(\mathcal{A}_i) = B$, and (ii) $\forall a \in \mathcal{A}_i$, $i \in f(a)$.
        Fix $1 \leq i < j \leq k$, and let $\mathcal{H}_{i,j} \subseteq H_{i,j}$ such that $\mathcal{A}_i \cap H_{i,j} = \mathcal{H}_{i,j}$,
        while $\mathcal{A}_j \cap H_{i,j} = H_{i,j} \setminus \mathcal{H}_{i,j}$.
        Note that for all $i \in [k]$, it holds that
        \begin{equation}\label{eq:fes_forward_direction}
            \sum_{j \in [i-1]} \big( \Sigma(H_{j,i}) - \Sigma(\mathcal{H}_{j,i}) \big) + \sum_{j \in [i+1,k]} \Sigma(\mathcal{H}_{i,j}) = B.
        \end{equation}
        Additionally, let $r_{i,j} = |\mathcal{S}(H_{i,j}) \cap [0,\Sigma(\mathcal{H}_{i,j})-1]|$ denote the number of subset sums of $H_{i,j}$
        of sum at most $\Sigma(\mathcal{H}_{i,j})-1$.
        Set $S = \setdef{\sigma^{i,j}_{\Sigma(\mathcal{H}_{i,j})}, \tau^{i,j}_{\Sigma(\mathcal{H}_{i,j})}, v^{i,j}_{2B - r_{i,j}}}{1 \leq i < j \leq k}$.
        It holds that $|S| = 3 \binom{k}{2} = p$.
        In the following, we will prove that every connected component of $G-S$ has weight at most $\ell$.

        Notice that the vertices of $\hat{C}_i$ and $\hat{C}_j$ are in different connected components of $G-S$,
        and let $\hat{\mathcal{C}}_i$ denote the connected component of $G-S$ that contains the vertices of $\hat{C}_i$, for all $i \in [k]$.
        Additionally, for every $1 \leq i < j \leq k$, let $\hat{\mathcal{P}}_{i,j}$ denote the connected component of $G-S$ that contains
        the vertices of $\hat{D}^3_{i,j}$.

        Fix $1 \leq i < j \leq k$.
        Notice that $S \cap \hat{D}_{i,j} = \braces{\sigma^{i,j}_{\Sigma(\mathcal{H}_{i,j})}, \tau^{i,j}_{\Sigma(\mathcal{H}_{i,j})}}$.
        Consequently, it holds that $\hat{\mathcal{P}}_{i,j}$ contains $\Sigma(H_{i,j})$ vertices of weight $1$,
        $|\mathcal{S}(H_{i,j})| - 1$ vertices of weight $M$, as well as vertices $b^{i,j}_1$, $b^{i,j}_2$, and $b^{i,j}_3$.
        Therefore, it follows that $w(\hat{\mathcal{P}}_{i,j}) = \Sigma(H_{i,j}) + (|\mathcal{S}(H_{i,j})|-1) \cdot M + L/2 + L/2 + (k-1) \cdot 2BM + B - \Sigma(H_{i,j}) - (|\mathcal{S}(H_{i,j})|-1) \cdot M$,
        thus $w(\hat{\mathcal{P}}_{i,j}) = L + (k-1) \cdot 2BM + B = \ell$.

        Additionally, it holds that $\hat{\mathcal{C}}_i$ contains $\Sigma(\mathcal{H}_{i,j})$ vertices of weight $1$ belonging to $\hat{D}_{i,j}$,
        as well as $r_{i,j} + 2B - r_{i,j} = 2B$ vertices of weight $M$ from $\hat{D}_{i,j} \cup \hat{U}_{i,j}$.
        As for $\hat{\mathcal{C}}_j$, it contains $\Sigma(H_{i,j}) - \Sigma(\mathcal{H}_{i,j})$ vertices of weight $1$ belonging to $\hat{D}_{i,j}$,
        as well as $|\mathcal{S}(H_{i,j})| - 1 - r_{i,j} + 4B - |\mathcal{S}(H_{i,j})| + 1 - (2B - r_{i,j}) = 2B$ vertices of weight $M$ from $\hat{D}_{i,j} \cup \hat{U}_{i,j}$.
        For any fixed $i \in [k]$, it follows that
        \begin{align*}
            w(\hat{\mathcal{C}}_i) &= L +
            \sum_{j \in [i-1]} \big( \Sigma(H_{i,j}) - \Sigma(\mathcal{H}_{i,j}) + 2BM \big) +
            \sum_{j \in [i+1,k]} \big( \Sigma(\mathcal{H}_{i,j}) + 2BM \big)\\
            &= L + (k-1) \cdot 2BM + B\\
            &= \ell,
        \end{align*}
        where the second equality is due to~\cref{eq:fes_forward_direction}.
        This concludes the proof.
    \end{nestedproof}

    \begin{lemma}\label{lemma:fes_opposite_direction}
        If $(G,w,\ell,p)$ is a Yes-instance of \SWCOC, then $(A,k,f)$ is a Yes-instance of \RUBP.
    \end{lemma}

    \begin{nestedproof}
        Let $S_0 \subseteq V(G)$ such that $|S_0| \leq p = 3 \binom{k}{2}$,
        and for every connected component of $G - S_0$ it holds that the sum of the weights of its vertices
        is at most $\ell = L + (k-1) \cdot 2BM + B$.
        The following claim shows that there exists $S \subseteq V(G)$ such that $S \cap \hat{C}_i = \emptyset$, for all $i \in [k]$;
        in fact $S$ contains either $1$ or $2$ vertices per path between cliques.

        \begin{claim}
            There exists a set $S \subseteq V(G)$ such that $|S| \leq p$,
            and for every connected component of $G - S$ it holds that the sum of the weights of its vertices
            is at most $\ell$.
            Additionally, for all $1 \leq i < j \leq k$,
            it holds that $|S \cap \hat{U}_{i,j}| = |S \cap \hat{D}^1_{i,j}| = |S \cap \hat{D}^2_{i,j}| = 1$.
        \end{claim}

        \begin{claimproof} % LipiCS
        % \begin{nestedproof} % LNCS
            To prove the statement we will first introduce the following reduction rule.

            \proofsubparagraph*{Rule $(\dagger)$.} % LipiCS
            % \paragraph*{Rule $(\dagger)$.} % LNCS
            Let $Z \subseteq V(G)$ be a deletion set such that $|Z \cap \hat{C}_i| \geq 2(k-1)$ for some $i \in [k]$,
            while every connected component of $G - Z$ has weight at most $\ell$.
            Then, replace $Z$ with $Z' = (Z \setminus \hat{C}_i) \cup N(\hat{C}_i)$,
            where $N(\hat{C}_i) = \bigcup_{v \in \hat{C}_i} N(v) \setminus \hat{C}_i$.

            It is easy to see that $|Z'| \leq |Z|$, since $|N(\hat{C}_i)| = 2(k-1)$.
            Moreover, it holds that every connected component of $G - Z'$ has weight at most $\ell$.
            To see this, it suffices to consider the connected component that contains the vertices of $\hat{C}_i \setminus Z$.
            Let $\hat{\mathcal{C}}_i$ denote the set of vertices of the connected component of $G-Z$ where $\hat{C}_i \setminus Z \subseteq \hat{\mathcal{C}}_i$,
            and $\hat{\mathcal{C}}'_i = \hat{\mathcal{C}}_i \setminus (\hat{C}_i \cup N(\hat{C}_i))$.
            It holds that $w(\hat{\mathcal{C}}'_i) \leq w(\hat{\mathcal{C}}_i) \leq \ell$,
            as well as $w(\hat{C}_i) = L \leq \ell$,
            while in $G-Z'$ the connected component $\hat{\mathcal{C}}_i$ is split into the connected component $\hat{C}_i$ as well as
            a partition of $\hat{\mathcal{C}}'_i$.

            Starting from $S_0$, let $S_1 \subseteq V(G)$ be the set obtained after applying Rule $(\dagger)$ exhaustively.
            Note that this process will finish in at most $k$ steps.
            Notice that it holds that $|S_1 \cap \hat{D}_{i,j}| \geq 1$, for all $1 \leq i < j \leq k$,
            since $w(\hat{D}_{i,j}) > \ell$.
            Moreover, observe that since $L/(4k) > (k-1) \cdot 2B M + B$,
            it follows that $\ell < L/(4k) \cdot (4k+1)$, thus at most $4k$ vertices of weight $L/(4k)$
            may be in the same connected component of $G - S_1$.
            In an analogous way, at most $2k$ vertices of weight $L/(4k)$ may be in the same connected component
            of $G - S_1$ with a vertex of weight $L/2$.
        
            Assume there exist $1 \leq i < j \leq k$ such that $S_1 \cap \hat{U}_{i,j} = \emptyset$.
            In that case, the vertices of $(\hat{C}_i \cup \hat{C}_j) \setminus S_1$ are in the same connected component of $G-S_1$,
            and since every such vertex is of weight $L/(4k)$, it follows that $|S_1 \cap (\hat{C}_i \cup \hat{C}_j)| \geq 4k$.
            Assume without loss of generality that $0 \leq |S_1 \cap \hat{C}_i| \leq |S_1 \cap \hat{C}_j| \leq 4k$,
            thus $|S_1 \cap \hat{C}_j| \geq 2k$ follows, which is a contradiction.

            Next, assume there exists $\hat{D}_{i,j}$ such that $|S_1 \cap \hat{D}_{i,j}| = 1$.
            In that case, it holds that either $b^{i,j}_1 \notin S_1$ or $b^{i,j}_3 \notin S_1$,
            and is in the same connected component of $G - S_1$ as the vertices of either $\hat{C}_i \setminus S_1$ or $\hat{C}_j \setminus S_1$.
            Assume without loss of generality that $b^{i,j}_1$ is in the same connected component of $G - S_1$
            as the vertices of $\hat{C}_i \setminus S_1$.
            Then, since $w(b^{i,j}_1) = L/2$ and every vertex of $\hat{C}_i$ is of weight $L / (4k)$,
            it follows that $|S_1 \cap \hat{C}_i| \geq 2k$, which is a contradiction.

            Consequently, for all $1 \leq i < j \leq k$ it holds that $|S_1 \cap \hat{D}_{i,j}| \geq 2$, as well as $S_1 \cap \hat{U}_{i,j} \neq \emptyset$.
            Since $|S_1| \leq 3 \binom{k}{2}$, it follows that $|S_1 \cap \hat{D}_{i,j}| = 2$, and $|S_1 \cap \hat{U}_{i,j}| = 1$.

            Notice that if $b^{i,j}_2 \in S_1$ or $S_1 \cap \hat{D}_{i,j} \subseteq \hat{D}^1_{i,j} \cup \braces{b^1_{i,j}}$ or
            $S_1 \cap \hat{D}_{i,j} \subseteq \hat{D}^2_{i,j} \cup \braces{b^3_{i,j}}$,
            then it follows once again that either $b^{i,j}_1 \notin S_1$ or $b^{i,j}_3 \notin S_1$,
            and is in the same connected component of $G - S_1$ as the vertices of either $\hat{C}_i \setminus S_1$ or $\hat{C}_j \setminus S_1$,
            leading to contradiction.

            Let $\hat{\mathcal{P}}_{i,j} \subseteq \hat{D}_{i,j}$ denote the connected component (path) of $G - S_1$ such that $b^2_{i,j} \in \hat{\mathcal{P}}_{i,j}$.
            Assume $b^{i,j}_1 \in S_1$ and consider $S_2 = (S_1 \setminus \braces{b^{i,j}_1}) \cup \braces{\sigma^{i,j}_{\Sigma(H_{i,j})}}$.
            Then, it follows that $\hat{\mathcal{P}}_{i,j} \subseteq (\braces{b^{i,j}_2, b^{i,j}_3} \cup \hat{D}^2_{i,j}) \setminus \braces{\tau^{i,j}_{\Sigma(H_{i,j})}}$,
            therefore $w(\hat{\mathcal{P}}_{i,j}) \leq \ell - L/2$.
            Consequently, every connected component of $G - S_2$ has weight at most $\ell$.
            Using analogous arguments, one can substitute $b^3_{i,j}$ with $\tau^{i,j}_0$ in case $b^3_{i,j}$ belongs to the deletion set.
            Let $S$ be the set obtained by those substitutions.
        % \end{nestedproof} % LNCS
        \end{claimproof} % LipiCS

        Let $\hat{\mathcal{C}}_i$ denote the connected component of $G - S$ containing the vertices of $\hat{C}_i$.
        Additionally, let $\hat{\mathcal{P}}_{i,j} \subseteq \hat{D}_{i,j}$ denote the connected component (path) of $G - S$ such that
        $\hat{D}^3_{i,j} \subseteq \hat{\mathcal{P}}_{i,j}$.

        \begin{claim}\label{claim:fes_size_of_comp}
            For all $1 \leq i < j \leq k$,
            it holds that $w(\hat{\mathcal{P}}_{i,j}) = \ell$,
            while both vertices of $S \cap \hat{D}_{i,j}$ are of weight $M$.
            Moreover, it holds that $w(\hat{\mathcal{C}}_i) = \ell$, for all $i \in [k]$.
        \end{claim}

        \begin{claimproof} % LipiCS
        % \begin{nestedproof} % LNCS
            Let $\hat{Q}_{i,j} = \hat{\mathcal{P}}_{i,j} \cup (S \cap \hat{D}_{i,j})$.
            Notice that since $\hat{D}^3_{i,j} \subseteq \hat{\mathcal{P}}_{i,j}$,
            it holds that the vertices of $S \cap \hat{D}_{i,j}$ are of weight at most $M$,
            thus $w (\hat{Q}_{i,j}) \leq \ell + 2M$.
            For the first statement, it suffices to prove that it holds $w (\hat{Q}_{i,j}) = \ell + 2M$.

            In the following we prove that $w (\hat{Q}_{i,j}) \geq \ell + 2M$.
            Notice that there are exactly $k$ additional connected components $\hat{\mathcal{C}}_1, \ldots, \hat{\mathcal{C}}_k$ in $G - S$,
            apart from all components $\hat{\mathcal{P}}_{i,j}$.
            Moreover, their weight sums up to
            \begin{equation}\label{eq:weight_of_C}                
                \sum_{i \in [k]} w(\hat{\mathcal{C}}_i) = w(G) - \sum_{1 \leq i < j \leq k} w(\hat{Q}_{i,j}) - \binom{k}{2} \cdot M,
            \end{equation}
            where the last term is due to the vertices in $S \cap \hat{U}_{i,j}$ which are of weight $M$,
            therefore
            \begin{equation}\label{eq:ugly_inequality}                
                w(G) - \sum_{1 \leq i < j \leq k} w(\hat{Q}_{i,j}) - \binom{k}{2} \cdot M \leq k \cdot \ell.
            \end{equation}
            In order to compute $w(G)$, notice that
            \begin{itemize}
                \item $w(\hat{C}_i) = L$, for all $i \in [k]$,
                \item $w(\hat{U}_{i,j}) = (4B - |\mathcal{S}(H_{i,j})| + 2) \cdot M$, for all $1 \leq i < j \leq k$,
                \item $w(\hat{D}_{i,j}) = 2 \Sigma(H_{i,j}) + 2 |\mathcal{S}(H_{i,j})| \cdot M + L + (k-1) \cdot 2BM + B - \Sigma(H_{i,j}) - (|\mathcal{S}(H_{i,j})| - 1) \cdot M$,
                for all $1 \leq i < j \leq k$,
            \end{itemize}
            and adding up the last two items gives
            \[
                w(\hat{U}_{i,j} \cup \hat{D}_{i,j}) = \Sigma(H_{i,j}) + L + B + (k-1) \cdot 2BM + (4B + 3) \cdot M.
            \]
            Consequently, it follows that
            \begin{align*}
                w(G) &=
                kL + kB + \binom{k}{2} \parens*{L + B + (k-1) \cdot 2BM + (4B + 3) \cdot M}\\
                &= k (L + B + (k-1) \cdot 2BM) + \binom{k}{2} \parens*{L + B + (k-1) \cdot 2BM + 3M}\\
                &= k \ell + \binom{k}{2} (\ell + 3M),
            \end{align*}
            which due to~\cref{eq:ugly_inequality} gives
            \begin{align*}
                k \ell + \binom{k}{2} (\ell + 3M)
                - \sum_{1 \leq i < j \leq k} w(\hat{Q}_{i,j}) - \binom{k}{2} \cdot M \leq k \cdot \ell,
            \end{align*}
            thus
            \[
                \binom{k}{2} (\ell + 2M) \leq \sum_{1 \leq i < j \leq k} w(\hat{Q}_{i,j}),
            \]
            and since $w (\hat{Q}_{i,j}) \leq \ell + 2M$,
            it follows that $w (\hat{Q}_{i,j}) = \ell + 2M$,
            for all $1 \leq i < j \leq k$.

            As for the weight of the components $\hat{\mathcal{C}}_i$,
            due to $w(\hat{\mathcal{C}}_i) \leq \ell$ and~\cref{eq:weight_of_C},
            it follows that $w(\hat{\mathcal{C}}_i) = \ell$ for all $i \in [k]$.
        % \end{nestedproof} % LNCS
        \end{claimproof} % LipiCS

        Due to~\cref{claim:fes_size_of_comp}, it follows that $S \cap \hat{D}_{i,j} = \braces{\sigma^{i,j}_q, \tau^{i,j}_q}$,
        for some $q \in \mathcal{S}(H_{i,j})$.
        Since $w(\hat{\mathcal{C}}_i) = \ell$, while $kB < M$, it follows that $\hat{\mathcal{C}}_i$ contains exactly $(k-1) \cdot 2B$ vertices of weight $M$,
        as well as exactly $B$ vertices of weight $1$.
        Let $(\mathcal{A}_1, \ldots, \mathcal{A}_k)$ be a partition of $A$ defined in the following way:
        for all $1 \leq i < j \leq k$, if $\sigma^{i,j}_q \in S$,
        then $\mathcal{A}_i \cap H_{i,j} = \mathcal{H}_{i,j}$ and $\mathcal{A}_j \cap H_{i,j} = H_{i,j} \setminus \mathcal{H}_{i,j}$,
        where $\mathcal{H}_{i,j} \subseteq H_{i,j}$ such that $\Sigma(\mathcal{H}_{i,j}) = q$.
        Notice that $\Sigma(\mathcal{A}_i)$ is equal to the number of vertices of weight $1$ in $\hat{\mathcal{C}}_i$,
        therefore $\Sigma(\mathcal{A}_i) = B$ follows.
    \end{nestedproof}

    \begin{lemma}\label{lemma:fes_parameter_bound}
        It holds that $\fes(G) = \bO (k^3)$ and $\Delta(G) = \bO(k)$.
    \end{lemma}

    \begin{nestedproof}
        Let $F \subseteq E(G)$ contain all edges between vertices of $\hat{C}_i$, for all $i \in [k]$,
        as well as all edges which are adjacent to the endpoints of the paths $\hat{U}_{i,j}$ and $\hat{D}_{i,j}$.
        Notice that
        \[
            |F| = k \cdot \binom{4k}{2} + 4\binom{k}{2} \cdot 4k
            = \bO(k^3),
        \]
        while the graph remaining after the deletion of the edges in $F$ is a forest.
        
        For the maximum degree, notice that $|N(v)| = 4k-1 + 2(k-1) = \bO(k)$, for all $v \in \bigcup_{i \in [k]} \hat{C}_i$.
        As for the vertices of the paths, they are all of degree $2$,
        apart from the endpoints which are neighbors with all the vertices of a single clique,
        therefore of degree $\bO(k)$.
    \end{nestedproof}

    Due to~\cref{lemma:fes_forward_direction,lemma:fes_opposite_direction,lemma:fes_parameter_bound}, the statement follows.
\end{proof}

By \cref{thm:swcoc_to_swvi,thm:SWCOC:fes}, the hardness of \SWVI{} follows.

\begin{theorem}\label{thm:SemiVI:fes}
    \SWVI{} is W[1]-hard parameterized by $\fes + \Delta$.
\end{theorem}

Moreover, we can easily reduce an instance $(G,w,k)$ of \SWVI{} to an instance $(G',k)$ of \UVI{}
by attaching a path on $w(v)-1$ vertices
to each vertex $v$ (we assume that $w(v) \leq k$,
otherwise $v$ belongs to the deletion set).
Thus, $\fes(G') = \fes(G)$ and $\Delta(G') = \Delta(G) + 1$,
and due to~\cref{thm:SemiVI:fes} the main result of this section follows.

\begin{theorem}\label{thm:UVI:fes}
    \UVI{}  is W[1]-hard parameterized by $\fes + \Delta$.
\end{theorem}

\section{Max-Leaf Number}

In this section, we consider \UVI{} parameterized by the max-leaf number.
For a connected graph $G$ we denote by $\ml(G)$ the
maximum number of leaves of any spanning tree of $G$. This is a well-studied
but very restricted parameter \cite{FellowsJR13,FellowsLMMRS09,Lampis12}. In
particular, it is known that if a graph $G$ has $\ml(G) \le k$, then in fact $G$
is a subdivision of a graph on $\bO(k)$ vertices~\cite{KleitmanW91}. We are
motivated to study this parameter because in a sense it lies close to the
intractability boundary established in \cref{sec:fes}. Observe
that if a graph is a sub-division of a graph on $k$ vertices, then it has 
maximum degree at most $k$ and feedback edge set at most $k^2$; however, graphs of small
feedback edge set and small degree do not necessarily have small max-leaf
number (consider a long path where we attach a leaf to each vertex).
Interestingly, the graphs we construct in~\cref{subsec:rubp_to_swcoc} \emph{do} have small
max-leaf number, if we consider semi-weighted instances.  However, adding the
necessary simple gadgets in order to simulate weights increases the max-leaf number of
the graphs of our reduction. It is thus a natural question whether
this is necessary. In this section, we show that indeed this is inevitable, as
\VI{} is FPT parameterized by $\ml$.

We start with a high-level overview of our approach.
As mentioned, we will rely on the result of Kleitman and West~\cite{KleitmanW91} who showed that if a
graph $G=(V,E)$ has $\ml(G)\le k$, then there exists a set $X$ of size
$|X|=\bO(k)$ such that all vertices of $V\setminus X$ have degree at most $2$.
Our main tool is a lemma (\cref{lem:ml}) which allows us to ``rotate''
solutions: whenever we have a cycle in our graph, we can, roughly speaking,
exchange every vertex of $S$ in the cycle with the next vertex, until we reach
a point where our solution removes strictly more vertices of $X$. We therefore
guess the largest intersection of an optimal separator with $X$, and can now
assume that in every remaining cycle, the separator $S$ is not using any
vertices. This allows us to simplify the graph in a way that removes all cycles
and reduces the case to a tree, which is polynomial-time solvable.

Let us now give more details. We first recall the result of~\cite{KleitmanW91}.

\begin{lemmarep}\label{lem:kw}
    In any graph $G$, the set $X$ of vertices of degree at least $3$ has size at most $|X| \leq 12\ml(G) + 32$.
\end{lemmarep}

\begin{proof}
    The statement is trivially true if the maximum degree of the graph is less than $3$,
    so in the following assume that this is not the case.
    Kleitman and West actually proved that in a graph $G$ with $n$ vertices and minimum degree $3$ or more,
    $\ml(G) \geq n/4$.
    Given a graph $G$, let $X_G$ and $Y_G$ denote the set of vertices of degree at least $3$
    and at most $2$ respectively, i.e.~$V(G) = X_G \cup Y_G$ where
    $X_G = \setdef{v \in V(G)}{\deg_G(v) \geq 3}$ and $Y_G = \setdef{v \in V(G)}{\deg_G(v) \leq 2}$.
    Start from an arbitrary (connected) graph $G$, and greedily contract any edge in $G$,
    as long as the number of vertices of degree at least $3$ is not reduced.
    Call the resulting graph $G'$, where $|X_{G'}| \geq |X_G|$.
    Notice that the contraction of an edge does not increase the max-leaf number,
    thus $\ml(G') \leq \ml(G)$.
    In case $Y_{G'} = \emptyset$ the statement immediately follows by the result of Kleitman and West,
    so in the following assume that $Y_{G'} \neq \emptyset$.
    Moreover, for all $v \in Y_{G'}$, it holds that $N_{G'} (v) \cap X_{G'} \neq \emptyset$;
    if that were not the case, then contracting an edge incident to $v$ does not decrease
    the degree of any vertex in $X_{G'}$, which is a contradiction since in $G'$ any such edge
    has already been contracted.

    \begin{claim}
        It holds that $|Y_{G'}| \leq \ml(G')$.
    \end{claim}

    \begin{claimproof}
        It suffices to prove that there exists a spanning tree of $G'$
        where every vertex of $Y_{G'}$ is a leaf.
        Let $v \in Y_{G'}$, and notice that if $\deg_{G'} (v) = 1$,
        then $v$ is a leaf in any spanning tree of $G'$.
        In the following, assume that $N_{G'} (v) = \braces{u_1, u_2}$,
        where $N_{G'} (v) \cap X_{G'} \neq \emptyset$.
        Moreover, we claim that $\braces{u_1, u_2} \in E(G')$.
        Assume otherwise, and notice that in that case,
        the contraction of any edge incident on $v$ does not change the degree of any other vertex,
        which is a contradiction.
        
        First consider the case where $u_1, u_2 \in X_{G'}$.
        Given a spanning tree of $G'$ that contains both edges $\braces{v, u_1}$ and $\braces{v, u_2}$,
        one can obtain another spanning tree that contains
        edges $\braces{v, u_1}$ and $\braces{u_1, u_2}$ instead.

        For the remaining case, assume without loss of generality that $u_1 \in Y_{G'}$ and $u_2 \in X_{G'}$.
        Then, given a spanning tree of $G'$ that contains the edges $\braces{v, u_1}$ and either
        $\braces{v, u_2}$ or $\braces{u_1,u_2}$,
        one can obtain another spanning tree that contains
        edges $\braces{v, u_2}$ and $\braces{u_1, u_2}$ instead.
        
        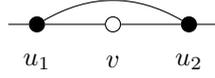
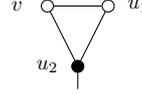
\begin{figure}[ht]
\centering 
  \begin{subfigure}[b]{0.4\textwidth}
  \centering
    \begin{tikzpicture}[scale=1, transform shape]
        \node[] (s) at (4.5,5) {};
        \node[black_vertex] (u1) at (5,5) {};
        \node[] () at (5,4.5) {$u_1$};
        \node[vertex] (v) at (6,5) {};
        \node[] () at (6,4.5) {$v$};
        \node[black_vertex] (u2) at (7,5) {};
        \node[] () at (7,4.5) {$u_2$};
        \node[] (t) at (7.5,5) {};
        \draw[] (s)--(u1)--(v)--(u2)--(t);
        \draw[] (u1) edge [bend left] (u2);
    \end{tikzpicture}
    \caption{Case 1.}
    \label{fig:ml_kw1}
  \end{subfigure}
\begin{subfigure}[b]{0.4\linewidth}
\centering
    \begin{tikzpicture}[scale=0.8, transform shape]
        \node[vertex] (v) at (5,5) {};
        \node[] () at (4.5,5) {$v$};
        \node[vertex] (u1) at (6,5) {};
        \node[] () at (6.5,5) {$u_1$};
        \node[black_vertex] (u2) at (5.5,4) {};
        \node[] () at (5,4) {$u_2$};
        \node[] (t) at (5.5,3.5) {};
        \draw[] (u2)--(v)--(u1)--(u2)--(t);
    \end{tikzpicture}
    \caption{Case 2.}
    \label{fig:ml_kw2}
  \end{subfigure}
\caption{Black vertices belong to $X_{G'}$.}
\end{figure}

        Consequently, there exists a spanning tree of $G'$ where every vertex of $Y_{G'}$ is a leaf,
        thus $|Y_{G'}| \leq \ml(G')$ follows.
    \end{claimproof}

    Lastly, we apply the result of Kleitman and West.
    If $|Y_{G'}| \geq 3$, add at most $|Y_{G'}| \leq \ml(G')$ edges connecting vertices of $Y_{G'}$,
    so that the resulting graph $G''$ has minimum degree at least $3$.
    Each such addition increases the max-leaf number by at most $2$,
    therefore it holds that $\ml(G'') \leq \ml(G') + 2\ml(G')$, while $|X_{G''}| = |X_{G'} \cup Y_{G'}| > |X_{G'}|$,
    and since $|X_{G''}| \leq 4\ml(G'')$ it follows that $|X_G| \leq 12\ml(G)$.
    It remains to consider the case where $1 \leq |Y_{G'}| \leq 2$.
    Then, it holds that $|X_{G'}| \geq 3$, since otherwise it follows that
    either $X_{G'} = \emptyset$ or $Y_{G'} = \emptyset$.
    Consequently, add at most $2$ edges per vertex of $Y_{G'}$, so that the resulting graph $G''$ has minimum
    degree $3$, therefore $\ml(G'') \leq \ml(G') + 8$,
    while $|X_{G''}| = |X_{G'} \cup Y_{G'}| > |X_{G'}|$,
    and since $|X_{G''}| \leq 4\ml(G'')$ it follows that $|X_G| \leq 4\ml(G) + 32$.
\end{proof}

We will solve \COC: for a given $\ell$ we want
to calculate the minimum number of vertices $p$ such that there exists a
separator $S$ of size at most $p$ with all components of $G - S$ having
size at most $\ell$. To obtain an algorithm for \UVI, we will try all possible
values of $\ell$ and select the solution which minimizes $\ell+p$.

Our main lemma is now the following:

\begin{lemmarep}\label{lem:ml}
    Let $G=(V,E)$ be a graph and $X$ be a set of
    vertices such that all vertices of $V\setminus X$ have degree at most $2$ in
    $G$.  For all positive integers $\ell,p$, if there exists a separator $S$ of
    size at most $p$ such that all components of $G - S$ have size at most
    $\ell$, then there exists such a separator $S$ that also satisfies the
    following property: for every cycle $C$ of $G$ with $C\cap X\neq \emptyset$ we
    either have $C\cap S=\emptyset$ or $C\cap X\cap S\neq \emptyset$.
\end{lemmarep}
    
\begin{proof}
    We will show that if we have a separator $S$ that does not satisfy the
    property, then we can obtain an equally good separator that contains strictly
    more elements of $X$. Applying this argument exhaustively will yield the lemma.
    
    Fix some separator $S$ of size at most $p$ such that all components of
    $G - S$ have size at most $\ell$. Suppose that we have a cycle $C$ in
    $G$ of length $t$, say $C=v_0,v_1,\ldots,v_{t-1},v_0$ such that $C$ contains
    some vertices of $X$ and some vertices of $S$, but no vertex that belongs to
    both $X$ and $S$. Suppose that $C\cap
    S=\{v_{i_0},v_{i_1},\ldots,v_{i_{r-1}}\}$, where $r=|C\cap S|$ and
    $i_0<i_1<\ldots<i_{r-1}$. We claim that the set $S' = (S\setminus
    C)\cup\{v_{i_0+1},v_{i_1+1},\ldots,v_{i_{r-1}+1}\}$ (where additions are modulo
    $t$) is a separator of the same size as $S$ which also leaves components of
    size at most $\ell$ in $G - S'$. Informally, what we are doing is
    replacing every vertex of $C\cap S$ with the next vertex in the cycle $C$. It
    is clear that $|S|=|S'|$. Furthermore, because $C\cap S\cap X=\emptyset$, all
    vertices of $C\cap S$ have degree $2$ in $G$ and one of their neighbors is in
    $S'$.  Let $D_j$ be the component of $G - S$ that is adjacent to
    $v_{i_j}$ and $v_{i_{j-1}}$ (where subtraction is done modulo $r$). We observe
    that $G - S'$ has the component $D_j\cup\{v_{i_j}\}\setminus\{
    v_{i_{j-1}+1}\}$, which has the same size as $D_j$. Here we are using the fact
    that $v_{i_j}$ has degree $2$ and its neighbor $v_{i_j+1}$ is in $S'$. Hence,
    this procedure does not increase the size of any component (in fact,
    if $C\cap S'\cap X \neq \emptyset$, the size of some components,
    which previously contained vertices of $C \cap S' \cap X$, is reduced,
    since such components are divided in $G-S'$).
    If the new separator $S'$ has
    $C\cap S'\cap X \neq \emptyset$, we have increased the intersection of the
    separator with $X$.  If not, we repeat this process.  Since at each step the
    intersection between the separator and $X$ cannot decrease, we eventually
    obtain a separator that satisfies the condition of the lemma.
\end{proof}
    
We are now ready to state the main result of this section.

\begin{theoremrep}\label{thm:ml}
    \UVI{} can be solved in time $2^{\bO(\ml)}n^{\bO(1)}$.
\end{theoremrep}

\begin{proof}
    We will solve \COC{} for a given maximum component size
    $\ell$.  As explained, we can reduce \UVI{} to this problem
    by trying out all values of $\ell$. We are given a graph $G$ and let $X$ be
    the set of vertices of degree at least $3$. By \cref{lem:kw}, $|X| \leq 12\ml(G) + 32$.
    Let $S$ be a separator of minimum size such that $G - S$ has only components
    of size at most $\ell$. Among all such separators, select an $S$ such that
    $S\cap X$ is maximized. Fix this separator $S$ for the analysis of the rest of
    the algorithm.
    
    Our algorithm as a first step guesses $S\cap X$, removes these vertices from
    the graph and sets $p:=p-|S\cap X|$. In the remainder, we focus on the case
    where this guess was correct. To simplify the rest of the presentation we
    assume that some vertices of the graph may be marked as undeleteable, meaning
    that we are only looking for separators that do not contain such vertices.
    Initially, the vertices of $X\setminus S$ are undeleteable.
    
    We now perform a series of polynomial-time simplification steps. First, if the current graph contains a connected component of size at most $\ell$, we remove
    this component. This is clearly correct. Second, if the graph contains a
    component where all vertices have degree at most $2$, we compute in polynomial
    time an optimal deletion set for this component, update $p$ appropriately,
    and remove this component.
    
    We are now in a situation where every cycle $C$ in our current graph must
    contain a vertex of $X$. By \cref{lem:ml} any such cycle $C$ must have $C\cap
    S=\emptyset$. We now contract all the vertices of $C$ into a single vertex and attach to this vertex $|C|-1$ leaves. We mark this new vertex as undeletable,
    that is, we will only look for separators $S$ that do not contain it. We also
    place this vertex in $X$. It is not hard to see that this transformation is
    correct, since if the optimal solution does not place any vertex of $C$ into
    $S$, we can replace $C$ with a single undeletable vertex which contributes
    $|C|$ to the size of its component.
    
    Performing the above exhaustively results in a graph that contains no cycles.
    The problem can now easily be solved optimally via dynamic programming in
    polynomial time.
\end{proof}

\section{Modular Width}

In this section we revisit an algorithm of~\cite{Gima2023} establishing
that \WVI{} can be solved in time
$2^{\bO(\mw)}n^{\bO(1)}$, on graphs of modular width $\mw$, but only if weights are
polynomially bounded in $n$ (or equivalently, if weights are given in unary).
It was left as an explicit open problem in~\cite{Gima2023} whether this
algorithm can be extended to the case where weights are given in binary and can
therefore be exponential in $n$. We resolve this problem positively, by showing
how the algorithm of \cite{Gima2023} can be modified to work also in this
case, without a large increase in its complexity.

The high-level idea of the algorithm of \cite{Gima2023} is to perform
dynamic programming to solve the related \textsc{Weighted Component Order
Connectivity} problem. In this problem we are given a target component weight
$\ell$ and a deletion budget $p$ and are asked if it is possible to delete from
the graph a set of vertices with total weight at most $p$ so that the maximum
weight of any remaining component is at most $\ell$. Using this algorithm as a
black box, we can then solve \textsc{Weighted Vertex Integrity} by iterating
over all possible values of $\ell$, between $1$ and the target vertex
integrity. If vertex weights are polynomially bounded, this requires a
polynomial number of iterations, giving the algorithm of \cite{Gima2023}.
However, if weights are given in binary, the target vertex integrity could be
exponential in $n$, so in general, it does not appear possible to guess the
weight of the heaviest component in an optimal solution.

Our contribution to the algorithm of \cite{Gima2023} is to observe that
for graphs of modular width $\mw$ the weight of the heaviest component may take
at most $2^{\bO(\mw)}n$ distinct possible values. Hence, for this parameter,
guessing the weight of the heaviest component in an optimal solution can be
done in FPT time. We can therefore plug in this result to the algorithm of~\cite{Gima2023}
to obtain an algorithm for \WVI{} with binary weights running in time $2^{\bO(\mw)}n^{\bO(1)}$.

Our observation is based on the following lemma.

\begin{lemmarep}\label{lem:mw}
    Let $G=(V,E)$ be an instance of \WVI.
    There exists an optimal solution using a separator $S$ such that for all connected
    components $D$ of $G - S$ and modules $M$ of $G$ we have one of the
    following: (i) $M\cap D=\emptyset$, (ii) $M\subseteq D$,
    or (iii) $D\subseteq M$.
\end{lemmarep}

\begin{proof}
    Let $S$ be a separator giving an optimal solution, that is a solution
    minimizing the sum of the weight of $S$ and the weight of the heaviest
    component of $G - S$. Our strategy is to show that if $S$ does not
    satisfy the condition of the lemma, then $S$ contains some vertex $z$ which in
    the terminology of \cite{Gima2023} is redundant, that is, $z$ has
    neighbors in at most one connected component of $G - S$. As observed in
    \cite{Gima2023}, this implies that $S\setminus \{z\}$ is a separator that
    gives a solution that is at least as good as that given by $S$, so we can
    remove $z$ from $S$. Repeating this exhaustively will produce a separator $S$
    that satisfies the conditions of the lemma.
    
    Suppose then that $S$ is a separator such that for some component $D$ of
    $G - S$ and some module $M$ we have $M\cap D\neq\emptyset$, $D\setminus
    M\neq \emptyset$, and $M\setminus D\neq \emptyset$. Let $x\in D\cap M$ and
    $y\in D\setminus M$ such that $x,y$ are adjacent (such $x,y$ must exist, since
    $D$ is connected). Since $M$ is a module and $y\not\in M$, we have that $y$ is
    adjacent to all of $M$.  Hence, $M\setminus S\subseteq D$ because all vertices
    of $M\setminus S$ are adjacent to $y\in D$.  It follows that each vertex $z\in
    M\setminus D$ must belong to the separator $S$.  We claim that $z$ is
    redundant, that is, $z$ only has neighbors in $D$ and in no other connected
    component of $G - S$.  This is not hard to see, since each neighbor of
    $z$ is either in $M$ (so is adjacent to $y\in D$) or is adjacent to all of $M$
    (hence also to $x$).
\end{proof}
    
We also recall the algorithmic result of \cite{Gima2023}.

\begin{theorem}\label{thm:mwold}(\cite{Gima2023})
    There exists an algorithm that takes as input a vertex-weighted graph $G=(V,E)$ and
    an integer $\ell$ and computes the minimum integer $p$ such that there
    exists a separator $S$ of $G$ of weight at most $p$ such that each component of $G - S$ has
    weight at most $\ell$. The algorithm runs in time $2^{\bO(\mw)}n^{\bO(1)}$, where
    $n$ is the size of the input.
\end{theorem}

Putting \cref{lem:mw,thm:mwold} together we obtain the main result
of this section.

\begin{theoremrep}\label{thm:mw}
    There exists an algorithm that solves
    \WVI{} in time $2^{\bO(\mw)}n^{\bO(1)}$, where $\mw$ is
    the modular width of the input graph, $n$ is the size of the input, and weights
    are allowed to be written in binary.
\end{theoremrep}
    
\begin{proof}
    Fix some optimal separator $S$ satisfying the conditions of \cref{lem:mw}. Our
    strategy is to use these conditions to guess the weight of the heaviest
    component of $G - S$. If this weight is $\ell$, then we can simply call
    the algorithm of \cref{thm:mwold}. More precisely, we will show that by
    considering $\bO(2^{\mw}n)$ distinct values, we are guaranteed to find
    $\ell$.  Hence, by executing the algorithm of \cref{thm:mwold}
    $\bO(2^{\mw}n)$ times and picking the best solution we find the optimal
    solution.
    
    Consider the modular decomposition tree associated with $G$,
    which can be computed in polynomial time~\cite{McConnellS99} and is
    recursively constructed as follows: if $G=(V,E)$ has at most $k$ vertices, then
    $G$ has a root labeled $V$ with $k$ children, each corresponding to a vertex of
    $V$; while if $|V|>k$, then $V$ can be partitioned into $k'\le k$ modules
    $V_1,\ldots, V_{k'}$, so we construct a root labeled $V$ and give it as
    children the roots of the trees constructed for $G[V_i]$, for $i\in[k']$. Each
    node of the constructed tree is labeled with a module of $G$.
    
    Let $D$ be the heaviest component of $G - S$. Find the node of the
    modular decomposition that is as far from the root as possible and satisfies
    that the corresponding module $M$ has $D\subseteq M$. (Such a node exists,
    since $D\subseteq V$.) The module $M$ can be decomposed into $k'\le k$ modules
    $M_1,\ldots, M_{k'}$. For each such module we have either $M_i\cap D=\emptyset$
    or $M_i\subseteq D$, by \cref{lem:mw} and because if we had $D\subseteq M_i$,
    this would contradict our choice of $M$. It must therefore be the case that for
    some $I\subseteq [k']$ we have $D=\bigcup_{i\in I}M_i$. 
    
    We now observe that we have $\bO(n)$ choices of $M$ (since the modular
    decomposition has $\bO(n)$ nodes), and $2^{\mw}$ choices for $I$. Hence, there
    are $\bO(2^{\mw}n)$ possible values of the weight of the heaviest component
    $\ell$.
\end{proof}

\section{Vertex Cover Number}

In this section, we design single-exponential algorithms for \VI{} parameterized by vertex cover number.
We suppose that a minimum vertex cover $C$ of size $\vc$ is given since it can be computed in time
$\bO(1.2738^{\vc} + {\vc} n)$~\cite{ChenKX10}.
% and in time $\bO(1.25284^{\vc} + {\vc} n)$~\cite{stacs/0001N24}.
We start by presenting an algorithm for \UVI,
before moving on to the weighted version of the problem.

% \subsection{Unweighted Case}

\begin{theoremrep}\label{thm:vc:unweighted}
    \UVI{} can be solved in time $5^{\vc}n^{\bO(1)}$.
\end{theoremrep}

\begin{proof}
    Let $C$ be a vertex cover of $G$ of size $|C| = \vc$.
    Without loss of generality, we assume that $k < \vc+1$, since otherwise $C$ is a
    separator that realizes the vertex integrity at most $k$.  For the analysis,
    fix an optimal separator $S$, which by \cref{prop:irredundant} can be assumed
    to be irredundant. We first guess $S \cap C$, delete its vertices, and decrease
    $k$ by $|S \cap C|$. In the remainder, assume that $S \cap C = \emptyset$. Let
    $I = V \setminus C$ be the independent set of the graph. We have to decide for
    each vertex of $I$ whether to place it in $S$ or not.
    
    We keep track of the vertices of $I$ which have been marked as belonging to $S$, denoted by $I_S$,
    those that have been marked as not belonging to $S$, denoted by $I_{\Bar{S}}$,
    and those which are undecided, denoted by $I_U$ (initially $I = I_U$).
    We also keep track of the connected components of the graph induced by $C \cup I_{\Bar{S}}$.
    Now, as long as there exists an undecided vertex $v \in I_U$
    such that all its neighbors in $C$ are in the same connected component,
    we mark that $v$ is not in $S$ and move it to $I_{\Bar{S}}$ (because $v$ would be redundant).
    Consider then an undecided $v \in I_U$ that has neighbors in $C$ which are in distinct components
    of the graph induced by $C \cup I_{\Bar{S}}$.
    We consider two cases: $v \in S$ and $v \notin S$.
    In the first case, we remove $v$ from the graph and decrease $k$ by $1$.
    In the latter, we observe that the number of connected components made up of vertices
    of $C \cup I_{\Bar{S}}$ has decreased.
    We continue this branching procedure until $k < 0$ (in which case we reject);
    or all vertices have been decided (in which case we evaluate the solution).
    The algorithm is correct, assuming that our guess of $S \cap C$ was correct,
    because for all vertices of $I$ we either know that we have made the correct choice
    (because $S$ is irredundant) or we consider both possible choices.
    
    For the running time, we define a potential function as the sum of $k$ plus the
    number of connected components induced by $C \cup I_{\Bar{S}}$.
    Then for both branches, the potential function decreases by at least 1.
    Moreover, the value of the potential function is bounded by $k - |C \cap S| + |C \setminus S|
    \leq 2\vc - 2|C \cap S|$.
    Note that the branching algorithm is applied to the graph after deleting $C \cap S$.
    Therefore, the running time is bounded by
    \begin{align*}
        \sum_{S \cap C \subseteq C} 2^{2\vc - 2|C\cap S|} n^{\bO(1)}&=
        \sum_{S\subseteq C} 4^{\vc - |C\cap S|}n^{\bO(1)} \\
        &= \sum_{S\subseteq C} 4^{\vc - |C\cap S|}1^{ |C\cap S|} n^{\bO(1)}\\
        &= \sum_{i=0}^\vc \binom{\vc}{i} 4^{\vc - i}1^{i}n^{\bO(1)} \\
        &= 5^\vc n^{\bO(1)},
    \end{align*}
    and the statement follows.
\end{proof}

% \subsection{Weighted Case}

We now move on to the weighted case of the problem. It is clear that
\WVI{} is FPT parameterized by vertex cover, due to~\cref{thm:mw} and
the relation between modular-width and vertex cover. However,
this gives a double-exponential dependence on $\vc$, as $\mw\le 2^{\vc}+\vc$
and there are some graphs for which this is essentially tight. We would like to
obtain an algorithm that is as efficient as that of \cref{thm:vc:unweighted}.
The algorithm of \cref{thm:vc:unweighted}, however, cannot be applied to the
weighted case because the case of the branching where we place a vertex of the
independent set in the separator is not guaranteed to make much progress (the
vertex could have very small weight compared to our budget).

Before we proceed, it is worth thinking a bit about how this can be avoided.
One way to obtain a faster FPT algorithm would be, rather than guessing only
the intersection of the optimal separator $S$ with the vertex cover $C$, to
also guess how the vertices of $C\setminus S$ are partitioned into connected
components in the optimal solution. This would immediately imply the decision
for all vertices of the independent set: vertices with neighbors in two
components must clearly belong to $S$, while the others cannot belong to $S$ if
$S$ is irredundant. This algorithm would give a complexity of
$\vc^{\bO(\vc)}n^{\bO(1)}$, however, because the number of partitions of $C$ is
slightly super-exponential.

Let us sketch the high level idea of how we handle this. Our first step is,
similarly to~\cref{thm:mw}, to calculate the weight $w_{\max}$ of the most
expensive connected component of the optimal solution. For this, there are at
most $2^{\vc}+n$ possibilities, because this component is either a single
vertex, or it has a non-empty intersection with $C$. However, if we fix its
intersection with $C$, then this fixes its intersection with the independent
set: the component must contain (by irredundancy) exactly those vertices of the
independent set all of whose neighbors in $C$ are contained in the component.
Having fixed a value of $w_{\max}$ we simply seek the best separator so that
all components have weight at most $w_{\max}$. The reason we perform this
guessing step is that this version of this problem is easier to decompose: if
we have a disconnected graph, we simply calculate the best separator in each
part and take the sum (this is not as clear for the initial version of \VI).

Suppose then that we have fixed $w_{\max}$, how do we find the best partition
of $C$ into connected components? We apply a win/win argument: if the optimal
partition has a connected component that contains many (say, more than
$\vc/10$) vertices of $C$, we simply guess the intersection of the component
with $C$ and complete it with vertices from the independent set, as previously,
while placing vertices with neighbors inside and outside the component in the
separator. If the weight of the component is at most $w_{\max}$, we recurse
in the remaining instance, which has vertex cover at most $9\vc/10$. The
complexity of this procedure works out as $T(\vc) \le 2^{\vc} \cdot T(9\vc/10) =
2^{\bO(\vc)}$.

What if the optimal partition of $C$ only has components with few vertices of
$C$? In that case we observe that we do not need to compute the full partition
(which would take time $\vc^{\vc}$), but it suffices to guess a good bipartition
of $C$ into two sets, of roughly the same size (say, both sets have size at
least $2\vc/5$), such that the two sets are a coarsening of the optimal
partition. In other words, we compute two subsets of $C$, of roughly equal
size, such that the intersection of each connected component with $C$ is
contained in one of the two sets. This is always possible in this case, because
no connected component has a very large intersection with $C$. Now, all vertices
of $I$ which have neighbors on both sides of the bipartition of $C$ must be
placed in the separator. But once we do this, we have disconnected the instance
into two independent instances, each of vertex cover at most $3\vc/5$.
The complexity of this procedure again works out as $T(\vc)\le 2^{\vc}\cdot2\cdot
T(3\vc/5) = 2^{\bO(\vc)}$.

Let us now proceed to the technical details.  To solve \WVI,
we first define the annotated and optimization version of the
problem.

\problemdef{\textsc{Annotated Weighted Vertex Integrity with Vertex Cover}}
{A vertex-weighted graph $G = (V, E, w)$, a vertex cover $C$ of $G$, an integer $w_{\max}$.}
{Find a minimum weight irredundant $\wvi$-set $S \subseteq V \setminus C$
such that $w(D) \le w_{\max}$  for all $D \in \cc(G-S)$.
If there is no such $S$, report NO.}

Then we give an algorithm that solves \textsc{Annotated Weighted Vertex
Integrity with Vertex Cover}.

\begin{theoremrep}\label{thm:annotatedWVI}
    {\sc Annotated Weighted Vertex Integrity with Vertex Cover} can be solved in time $2^{\bO(|C|)}n^{\bO(1)}$.
\end{theoremrep}

\begin{proof}
    Let $\mathcal{I} = (G, C, w_{\max})$ be an instance of \textsc{Annotated Weighted Vertex Integrity
    with Vertex Cover}, where $\OPT(\mathcal{I})$ denotes the weight of the optimal solution $S$,
    provided such a set $S$ exists.
    Additionally, let $k = |C|$ and $I = V \setminus C$.

    Assume that $N(v) = \emptyset$ for some $v \in I$.
    We claim that if $w(v) > w_{\max}$,
    then there is no irredundant $\wvi$-set $S \subseteq V \setminus C$
    such that $w(D) \le w_{\max}$  for all $D \in \cc(G-S)$ and we report NO.
    Assume there existed such a set $S$.
    Then, since $w(v) > w_{\max}$, it follows that $v \in S$.
    However, since $N(v) = \emptyset$, it follows that no component of $G-S$ contains neighbors of $v$,
    thus $v$ is redundant, which is a contradiction.
    If on the other hand it holds that $w(v) \leq w_{\max}$,
    then reduce $\mathcal{I}$ to $\mathcal{I}' = (G - v, C, w_{\max})$,
    where $\OPT(\mathcal{I}) = \OPT(\mathcal{I}')$.
    The correctness of the reduction rule is easy to see.
    In the following, assume $N(v) \neq \emptyset$, for all $v \in I$.
    We perform branching on whether there exists a component $D$ in $G -S$ such that $|D \cap C| \ge k/10$ or not.
    
    In the first case, we guess $D \cap C$.
    Since $S$ is irredundant,
    it holds that $D \setminus C = \setdef{v \in I}{N(v) \subseteq D \cap C}$.
    Moreover, let $I_R \subseteq I$ be the set of vertices with neighbors in both $D \cap C$ and $C \setminus D$.
    If $D$ is not connected or $w(D) > w_{\max}$, we immediately reject.
    Otherwise, notice that all vertices $v \in I_R$ must belong to $S$,
    thus we obtain a smaller instance $\mathcal{I}' = (G - (D \cup I_R), C \setminus D, w_{\max})$,
    where $|C \setminus D| \leq 9k/10$ and $\OPT(\mathcal{I}) = w(I_R) + \OPT(\mathcal{I}')$.

    In the second case, i.e.~when for all components $D \in \cc(G-S)$ it holds that $|D \cap C| < k/10$,
    we first show the following claim.
    \begin{claim}
        There exists a set $A \subseteq C$ of size $k/2 \le |A| \le 3k/5$
        such that every component $D$ of $G - S$ satisfies either
        $D \cap C \subseteq A$ or $D \cap C \subseteq C \setminus A$.
    \end{claim}
    
    \begin{claimproof} % LipiCS
    % \begin{nestedproof} % LNCS
        Let $\mathcal{D} = \{D_1, \ldots, D_p\}$ be the set of connected components of $G-S$.
        Note that $\bigcup_{D \in \mathcal{D}} (D \cap C) = C$ since $S \cap C = \emptyset$.
        By the assumption, $|D \cap C| < k/10$ holds for all $D \in \mathcal{D}$.
        Set $A = \bigcup_{i=1}^{j} (D_i \cap C)$ where $j$ is the maximum index such that $|A \setminus D_j| < k/2$.
        By the definition of $A$, $|A| \geq k/2$ holds.
        Moreover, since $|D_j| < k/10$, $|A| \leq k/2 + k/10 = 3k/5$.
        Thus, the claim holds.
    \end{claimproof} % LipiCS
    % \end{nestedproof} % LNCS
    
    Next, we guess $A \subseteq C$ by considering all possible subsets of $C$ whose sizes
    are between $k/2$ and $3k/5$.
    Since every component $D$ of $G - S$ satisfies either $D \cap C \subseteq A$ or
    $D \cap C \subseteq C \setminus A$,
    all vertices of $I_R \subseteq I$ must belong to $S$,
    where $v \in I_R$ if it has neighbors in both $A$ and $C \setminus A$.
    In that case, we obtain two separate instances $\mathcal{I}_1 = (G[A \cup I_A], A, w_{\max})$ and
    $\mathcal{I}_2 = (G[(C \setminus A) \cup I_{C \setminus A}], C \setminus A, w_{\max})$,
    where $I_Z = \setdef{v \in I}{N(v) \subseteq Z}$ for $Z \in \{A, C \setminus A\}$,
    and $\OPT(\mathcal{I}) = w(I_R) + \OPT(\mathcal{I}_1) + \OPT(\mathcal{I}_2)$.
    
    Consequently, for both cases we obtain smaller instances and hence,
    we only have to recursively compute the annotated problems.
    Finally, we analyze the running time of our algorithm.
    Let $T(k)$ denote the running time when $|C| = k$.
    Then it holds that $T(k) \le 2^k n^{\bO(1)} \cdot T(9k/10) + 2^k n^{\bO(1)} \cdot 2T(3k/5) + n^{\bO(1)}$,
    thus $T(k) = 2^{\bO(k)} n^{\bO(1)}$ follows.
\end{proof}

\begin{theoremrep}\label{thm:vc:weighted}
    \WVI{} can be solved in time $2^{\bO(\vc)}n^{\bO(1)}$.
\end{theoremrep}

\begin{proof}
    Let $C$ be a vertex cover of $G$ of size $|C| = \vc$ and $I = V \setminus C$.
    For the analysis, fix an optimal separator $S$, which by \cref{prop:irredundant} can be assumed
    to be irredundant.
    We first guess $S \cap C$,
    and set $C' = C \setminus S$.
    In the following, let for $v \in I$, $N(v)$ denote the neighborhood of $v$ in $G$,
    while $N'(v) = N(v) \cap C'$.
    Let $D_{\max}$ denote a connected component of $G-S$ of maximum weight $w(D_{\max}) = w_{\max}$,
    and set $D_{\max} \cap C' = A \subseteq C'$.
    We argue that given $S \cap C$, there are at most $2^{|C'|} + n \leq 2^\vc + n$ cases regarding $D_{\max}$,
    where $n = |V|$.
    If $A = \emptyset$,
    then $D_{\max} = \{v\}$ for some vertex $v \in I$ with $N(v) \subseteq S \cap C$,
    for a total of at most $|I| \leq n$ choices.
    Otherwise, it holds that $A \neq \emptyset$, and due to the irredundancy of $S$ it follows that
    $D_{\max} \setminus A = \setdef{v \in I}{\emptyset \neq N'(v) \subseteq A}$,
    while any vertex of
    $I_R = \setdef{v \in I}{N'(v) \cap A, N'(v) \cap (C' \setminus A) \neq \emptyset}$ must belong to $S$.

    Consequently, it suffices to guess $S \cap C$ and $D_{\max}$,
    assure that $D_{\max}$ is indeed connected as intended, and then compute the expression
    \[
        \OPT(\mathcal{I}) + w(S \cap C) + w(I_R) + w_{\max},
    \]
    where $\mathcal{I} = (G - (D_{\max} \cup I_R \cup (S \cap C)), C', w_{\max})$ is an instance of
    \textsc{Annotated Weighted Vertex Integrity with Vertex Cover} and $\OPT(\mathcal{I})$
    the weight of its optimal solution.
    Due to \cref{thm:annotatedWVI}, $\operatorname{OPT}(\mathcal{I})$ can be computed
    in time $2^{\bO(\vc)} n^{\bO(1)}$,
    therefore the total running time is
    $2^\vc \cdot (2^\vc + n) \cdot n^{\bO(1)} \cdot 2^{\bO(\vc)} n^{\bO(1)} = 2^{\bO(\vc)} n^{\bO(1)}$.
\end{proof}

\section{Conclusion}

We have presented a number of new results on the parameterized complexity of
computing vertex integrity. The main question that remains open is whether the
slightly super-exponential $k^{\bO(k)} n^{\bO(1)}$ algorithm, where $k$ is the
vertex integrity itself, can be improved to single-exponential. Although we
have given such an algorithm for the more restricted parameter vertex cover, we
conjecture that for vertex integrity the answer is negative. Complementing this
question, it would be interesting to consider approximation algorithms for
vertex integrity, whether trying to obtain FPT approximations in cases where
the problem is W[1]-hard, or trying to obtain almost-optimal solutions via
algorithms that run with a better parameter dependence. Again, a
constant-factor or even $(1+\varepsilon)$-approximation running in time
$2^{\bO(k)} n^{\bO(1)}$ would be the ideal goal. Do such algorithms exist or can
they be ruled out under standard complexity assumptions?

%%
%% Bibliography
%%

%% Please use bibtex, 

\bibliography{bibliography}

% \appendix

\end{document}